\documentclass[11pt]{article}

\usepackage[vmargin=1.18in,hmargin=1.1in, tmargin=0.63in]{geometry}
\usepackage{amsmath,amssymb,amsfonts,graphicx,amsthm,nicefrac,mathtools,bm}
\usepackage[hidelinks]{hyperref}
\usepackage[authoryear,sort]{natbib}
\usepackage[english]{babel}
\usepackage[T1]{fontenc}
\usepackage{bbm,mdframed}

\usepackage{xcolor}

\usepackage{xspace}
\usepackage{rotating,multirow}
\usepackage{enumerate,graphics,enumitem}
\usepackage{subcaption}
\usepackage[algo2e]{algorithm2e}
\usepackage[export]{adjustbox}

\setlist[itemize]{noitemsep, topsep=0pt}
\setlist[enumerate]{itemsep=5pt, topsep=5pt, leftmargin=25pt}

\definecolor{DarkRed}{rgb}{0.5,0.1,0.1}
\definecolor{DarkBlue}{rgb}{0.1,0.1,0.5}

\newtheorem{theorem}{Theorem}

\definecolor{verylightblue}{rgb}{0.7,0.8,1}
  {\begin{mdframed}[backgroundcolor=verylightblue]\begin{theorem}}%
  {\end{theorem}\end{mdframed}}

\definecolor{verylightgray}{gray}{0.95}
  {\begin{mdframed}[backgroundcolor=verylightgray]\begin{proof}}%
  {\end{proof}\end{mdframed}}

\newtheorem{lemma}{Lemma}

\definecolor{verylightred}{rgb}{1,0.8,0.8}
  {\begin{mdframed}[backgroundcolor=verylightred]\begin{lemma}}%
  {\end{lemma}\end{mdframed}}

\newtheorem{proposition}{Proposition}
  {\begin{mdframed}[backgroundcolor=verylightblue]\begin{proposition}}%
  {\end{proposition}\end{mdframed}}

\newtheorem{definition}{Definition}
\theoremstyle{remark}

\makeatletter

\newtheorem*{rep@theorem}{\rep@title}
\newcommand{\newreptheorem}[2]
{\newenvironment{rep#1}[1]
{\def\rep@title{#2 \ref{##1}} \begin{rep@theorem}}%
 {\end{rep@theorem}}}
\makeatother
\newreptheorem{theorem}{Theorem}
\newreptheorem{lemma}{Lemma}
\newreptheorem{corollary}{Corollary}
\newreptheorem{proposition}{Proposition}

\DeclareMathOperator*{\argmax}{arg\,max}
\DeclareMathOperator*{\argmin}{arg\,min}

\newcommand{\PP}[1]{\mathbb{P}\!\left\{{#1}\right\}} 
\newcommand{\E}{\mathbb{E}}

\newcommand{\PPst}[2]{\mathbb{P}\!\left\{{#1}\ \middle| \ {#2}\right\}} 

\def\O{\mathcal{O}}

\newcommand{\ignore}[1]{}

\newcommand{\eps}{\epsilon}

\let\emptyset\varnothing

\newcommand{\A}{\mathcal{A}}

\date{\thedate}
\author{\theauthor}
\title{\thetitle}

\usepackage[mmddyy]{datetime}

\def\P{\mathcal{P}}

\def\D{\mathcal{D}}

\def\G{\mathcal{G}}

\def\ci{\mathrm{CI}}

\makeatletter
\newcommand{\dotfrac}[2]{
\mathchoice
{\ooalign{$\genfrac{}{}{0pt}{0}{#1}{#2}$\cr\leavevmode\cleaders\hb@xt@ .22em{\hss $\displaystyle\cdot$\hss}\hfill\kern\z@\cr}}
{\ooalign{$\genfrac{}{}{0pt}{1}{#1}{#2}$\cr\leavevmode\cleaders\hb@xt@ .22em{\hss $\textstyle\cdot$\hss}\hfill\kern\z@\cr}}
{\ooalign{$\genfrac{}{}{0pt}{2}{#1}{#2}$\cr\leavevmode\cleaders\hb@xt@ .22em{\hss $\scriptstyle\cdot$\hss}\hfill\kern\z@\cr}}
{\ooalign{$\genfrac{}{}{0pt}{3}{#1}{#2}$\cr\leavevmode\cleaders\hb@xt@ .22em{\hss $\scriptscriptstyle\cdot$\hss}\hfill\kern\z@\cr}}
}
\makeatother

\usepackage{algorithm,algorithmic}

\makeatletter
\long\def\@makecaption#1#2{
        \vskip 0.8ex
        \setbox\@tempboxa\hbox{\small {\bf #1:} #2}
        \parindent 1.5em  
        \dimen0=\hsize
        \advance\dimen0 by -3em
        \ifdim \wd\@tempboxa >\dimen0
                \hbox to \hsize{
                        \parindent 0em
                        \hfil 
                        \parbox{\dimen0}{\def\baselinestretch{0.96}\small
                                {\bf #1.} #2
                                } 
                        \hfil}
        \else \hbox to \hsize{\hfil \box\@tempboxa \hfil}
        \fi
        }
\makeatother

\usepackage[algo2e]{algorithm2e}

\newcommand{\thedate}{}
\newcommand{\theauthor}{Paula Gradu*$^1$,~ Tijana Zrnic*$^1$,~ Yixin Wang$^2$,~ Michael I. Jordan$^1$\\ \\
$^1$University of California, Berkeley\\
$^2$University of Michigan}
\newcommand{\thetitle}{Valid Inference After Causal Discovery}

\usepackage{cleveref}

\usepackage{times}

\begin{document}

\maketitle

\begin{abstract}
Causal discovery and causal effect estimation are two fundamental tasks in causal inference. While many methods have been developed for each task individually, statistical challenges arise when applying these methods jointly: estimating causal effects after running causal discovery algorithms on the same data leads to ``double dipping,'' invalidating the coverage guarantees of classical confidence intervals. To this end, we develop tools for valid post-causal-discovery inference.
Across empirical studies, we show that a naive combination of causal discovery and subsequent inference algorithms leads to highly inflated miscoverage rates; on the other hand, applying our method provides reliable coverage while achieving more accurate causal discovery than data splitting.
\end{abstract}

\section{Introduction}
\label{sec:intro}

\emph{Causal discovery} and \emph{causal estimation} are
fundamental tasks in causal reasoning and decision-making. Causal
discovery aims to identify the underlying structure of the causal
problem, often in the form of a graphical representation that makes
explicit which variables causally influence which other variables, while causal estimation
aims to quantify the magnitude of the effect of one variable on another.
These two goals frequently go hand in hand: quantifying causal
effects requires adjustments that rely on either assuming or
discovering the underlying graphical structure.

Methodologies for causal discovery and causal estimation have mostly
been developed separately, and the statistical challenges that arise when
solving these problems jointly have largely been overlooked.
Indeed, a naive black-box combination of causal discovery
algorithms and standard inference methods for causal effects suffers
from ``double dipping.''  That is, classical confidence intervals, such as those
used for linear regression coefficients, need no longer cover the
target estimand if the causal structure is not fixed a priori but is
estimated on the same data used to compute the intervals.

Consider the following example. Suppose we collect a data set with
measurements of completely independent variables.  Since all the variables
are independent, the effect of any variable on another variable is of zero magnitude. However, if
the number of measured variables is sufficiently large and the sample
size is finite, it is likely that, purely by chance, there will be two
variables that seem sufficiently correlated, leading the causal
discovery algorithm to believe there is a causal link between
them. The problem is then compounded---the estimated effect along
this link will likely be biased away from zero because, after all, that is what drew
the algorithm to assert the existence of a causal link in the first place.

More generally, \emph{asserting the existence of a causal relationship
biases the estimated effect size toward significance}.
This phenomenon, whereby model selection can lead to a seemingly
significant relationship between a predictor and an outcome even if
they are perfectly independent, is often known as Freedman's
paradox~\citep{freedman1983note}.



More formally, suppose we are given a fixed causal graph $G$. Let
$\beta_{G}$ denote a causal parameter of
interest within $G$, which will typically correspond to an effect of
one variable on another. Standard statistical methods take a data set
$\D$ and produce a confidence interval $\ci_G(\alpha;\D)$ such that
\begin{equation}
\label{eq:classical_validity}
\PP{\beta_{G} \not\in \ci_G(\alpha;\D)} \leq \alpha,
\end{equation}
where $\alpha\in(0,1)$ is a pre-specified error level. For example, if
the variables in $G$ follow a normal distribution, $\ci_G(\alpha;\D)$
can be obtained via classical t-statistics. However, if we
\emph{estimate} the causal graph $\widehat G$ from $\D$, this
guarantee breaks down; that is,
there is \emph{no} guarantee that $\PP{\beta_{\widehat G}
\not\in \ci_{\widehat G}(\alpha;\D)} \leq \alpha.$
This issue arises due to the coupling between the estimand
$\beta_{\widehat G}$ and the data used for inference, since $\widehat
G$ implicitly depends on $\D$.

To address this failure of naive inference, we develop tools for valid
statistical inference after causal discovery. We build on concepts introduced
in the literature on adaptive data analysis~\citep{dwork2015preserving,
dwork2015generalization} and post-selection inference \citep{berk2013valid} and develop causal discovery algorithms that
allow the computing of downstream confidence intervals with rigorous coverage
guarantees. Our key observation is that \emph{randomizing} causal discovery mitigates
the bias due to data reuse.
In particular, we show that, for a level $\tilde \alpha\leq \alpha$
depending on the level of randomization, naive intervals in the
sense of Eq.~\eqref{eq:classical_validity} satisfy $$\PP{\beta_{\widehat G}
\not\in \ci_{\widehat G}(\tilde \alpha;\D)} \leq \alpha,$$ where
$\widehat G$ is a causal structure estimated via a noisy
causal discovery algorithm. Randomization leads to a quantifiable
tradeoff between the quality of the discovered causal model and the
statistical power of downstream inferences: higher levels of
randomization imply lower model quality, but at the same time
allow tighter confidence intervals; that is, $\tilde \alpha$ is not much
smaller than the target error level $\alpha$. Moreover, we show
empirically that the proposed randomization schemes are not vacuous:
classical confidence intervals for causal effects indeed vastly
undercover the target causal effect when computed after model selection based on standard, noiseless causal discovery algorithms.



\subsection{Related Work}


There has been steady progress in
providing formal statistical guarantees for causal discovery and
causal estimation~\citep[see, e.g.,][]{maathuis2009estimating,maathuis2010predicting,nandy2018high,maathuis2021graphical}.
Most existing work, however, deals with statistical uncertainty
arising from the two stages of causal reasoning separately, an
approach that leads to Freedman's paradox. Indeed, when discussing estimation via the IDA
algorithm~\citep{maathuis2009estimating,maathuis2010predicting},
\citet{witte2020efficient} acknowledge this shortcoming of existing
tools, saying: ``\emph{When the graph is estimated on the same data as
used for IDA, the naive standard errors from the adjusted linear
regressions are invalid. Although considerable progress has been made
in the area of post-selection inference [\dots], no method has been
proposed specifically for estimating standard errors of causal effect
estimates after causal search.}''

In this work, we focus precisely on this challenge, aiming to provide
valid statistical inference after \emph{score-based} graph selection.
Our scope includes exhaustive scoring of all considered graphs and
selection of the one with the top score, as well as greedy equivalence
search (GES)~\citep{meek1997graphical, chickering2002optimal}. 
Although the
focus in this work is on score-based selections, we believe that many
of the principles we will introduce can be extended to other graph
estimation strategies.

The technical tools in our work build upon those introduced in the
literature on \emph{differential privacy}~\citep{dwork2006calibrating}
and \emph{adaptive data
analysis}~\citep{dwork2015preserving, dwork2015generalization, bassily2016algorithmic,
jung2019new}. The core idea in adaptive data analysis is to use
randomization as a means of mitigating overfitting that arises from
double dipping. In particular, we rely on the concept of
\emph{max-information}---first introduced by
\citet{dwork2015generalization} and subsequently studied by \citet{rogers2016max}---and its relationship to differential
privacy.


Our work is also closely related to work on \emph{post-selection
inference}, but the specific tools developed in the existing
literature fall short of solving the casual inference
problems that are our focus. Indeed, 
existing solutions for achieving validity in the
presence of selection are generally either simultaneous over all
possible selections~\citep[e.g.,][]{berk2013valid,
bachoc2020uniformly, kuchibhotlavalid} or require a tractable
characterization of possible selection
events~\citep[e.g.,][]{fithian2014optimal, lee2014exact, lee2016exact,
tibshirani2016exact}). In our problem setting, the former approach
would be highly conservative statistically, especially when the number
of considered graphs is large, and would generally be
intractable computationally except when the number of variables is small. The latter
approach is restricted to selection algorithms that admit an explicit
characterization of the data conditional on a given selection, and can typically be applied only to parametric data distributions. For
causal discovery, the former would require understanding the data
distribution conditional on the graph that was selected. Given the
complexity of graph estimation algorithms, it is not clear how to
obtain such a characterization for popular causal discovery
algorithms. Furthermore, our correction principle is entirely nonparametric. Within post-selection
inference, our work is most closely related to a thread of research that involves randomizing selection
rules~\citep{tian2018selective,zrnic2020post, rasines2021splitting, leiner2022data, neufeld2023data}.


\section{Problem Formulation and Preliminaries}
\label{sec:preliminaries}

We formalize the problem of inference after
causal discovery and discuss the meaning of statistical
validity in this context. \Cref{sec:graph_prelims} revisits standard causal concepts. \Cref{sec:proj_targets} discusses targets of inference in the context of causal graphs. \Cref{sec:rnd_targets} formalizes
what constitutes valid inference and provides a high-level
overview of our randomization-based approach. \Cref{sec:ada} overviews
the key results from adaptive data analysis that lie at the core of
the subsequently developed algorithms.

\subsection{Causal Preliminaries}
\label{sec:graph_prelims}

We consider the problem of performing inference based on a causal graph. A
\emph{causal} graph is a directed acyclic graph (DAG) $G=(V,E)$, where
$V=(X_1, \ldots, X_d)$ is the set of vertices and $E$ is
the set of edges. We denote by $\mathbf{Pa}_j^G\subseteq[d]$ the set
of parents of node $X_j$ in graph $G$. In addition to capturing
conditional independence relationships, a causal graph represents the
causal relations in the data: the existence of an edge from $X_i$ to
$X_j$ implies a possible causal effect from $X_i$ to $X_j$.

Our theory also applies to causal discovery
methods that return an \emph{equivalence class} of DAGs, namely
a \emph{completed partially directed acyclic graph} (CPDAG). A CPDAG
is a partially directed graph with the same skeleton as the graphs in
the equivalence class, where directed edges represent arrows that are
common to all DAGs in the equivalence class, and the undirected edges
correspond to edges that are directed one way in some DAGs and the
other way in other DAGs in the equivalence class. We will use the notation $G$, as well as the term ``causal graph,'' to
refer to both DAGs and CPDAGs, given that our tools are largely agnostic to whether the
causal discovery criterion is applied to a set of possible DAGs or
CPDAGs.

\subsection{Targets of Inference in Causal Graphs}
\label{sec:proj_targets}

What makes post-selection inference conceptually challenging is the specification and interpretation of the inferential target. Indeed, the statistician may arrive at different causal graphs under different realizations of the data and, crucially, different graphs lead to different causal questions, different adjustment sets, and different identification formulas, implying in turn different \emph{targets of inference} in different graphs. Here, a ``target of inference'' is the population-level quantity that standard causal estimators aim to approximate. It is thus natural ask whether inference---and specifically its target---is meaningful if the discovered graph is not the graph underlying the data-generating process.

One perspective that resolves this issue is the view that different models provide different \emph{approximations} to the truth, some better than others, and should not be thought of as true data-generating processes \citep{berk2013valid, buja2019models1, buja2019models2}. We build upon this perspective in this work, accepting that although a causal graph is rarely a perfect representation of the truth, it can nevertheless serve as a useful working model. For instance, given the complexity of any real-world system, some relevant factors will almost inevitably be missing from the graph used in the analysis. This is true not only when the graph is estimated algorithmically, but also when it is provided by a domain expert.

Treating models as approximations leads to the following practical way of conceptualizing targets of inference. Fix the causal estimator that the statistician wishes to use once they have a causal graph (e.g., least-squares regression with a backdoor adjustment chosen based on the graph). Then, whether or not the graph is correct, there is an underlying population-level quantity that the estimator approximates, typically equal to its large-sample limit (which we assume always exists for simplicity). In the least-squares example, the target of inference is given by:
\begin{equation}
\label{eq:causal_effect}
\beta_G^{(i\rightarrow j)} = \left(\argmin_{\beta} \E_{(X_1,\dots,X_d)\sim\P} \left(X_j - \sum_{s \in A_G^{(i\rightarrow j)}\cup i} \beta_s X_s\right)^2\right)_{X_i},
\end{equation}
where $\P$ is the underlying data distribution, $A_G^{(i\rightarrow j)}\subseteq [d]$ is a valid adjustment set
in $G$, meaning that conditioning on $X_A$ blocks all backdoor paths
from $X_i$ to $X_j$ \citep{pearl2009causality}, and the subscript $X_i$ of the outermost
parentheses denotes taking the coefficient corresponding to
$X_i$. The parameter $\beta_G^{(i\rightarrow j)}$ exactly answers the causal
query $\frac{\partial}{\partial x}\E[{X_j \,|\,
\mathrm{do}(X_i=x)}]$ when $\P$ is a
multivariate normal distribution and $G$ is the true underlying
DAG. However, even when $G$ is not the true DAG, $\beta_G^{(i\rightarrow j)}$ is a meaningful target as it can be seen as a ``projection'' of the true data-generating process onto the working model $G$ with linear functional relationships. In general we will use $\beta_G^{(i\rightarrow j)}$ to denote the target of inference in graph $G$ when the statistician asks for the effect of $X_i$ on $X_j$, relying on some estimation strategy.


This perspective---closely related to the concept of regression functionals \citep{buja2019models2}---allows us to talk about valid statistical inference, regardless of whether the working causal graph is perfect or the functional form of the structural relationships among variables is well specified. This is true because the large-sample limit of a causal estimator can be defined for \emph{any} input graph.


We note that, when $G$ is a CPDAG, the target of inference
$\beta_G^{(i\rightarrow j)}$ should typically be thought of as
denoting a set of targets for each DAG in the equivalence class.


\subsection{Statistical Validity}
\label{sec:rnd_targets}

To perform a causal analysis, we work with a data set, $\D =
\{X^{(k)}\}_{k=1}^n \equiv \{(X_1^{(k)},\dots,X^{(k)}_d)\}_{k=1}^n$,
of~$n$ data points drawn independently from a distribution $\P$, where $X_j^{(k)}$
denotes the $j$-th variable in data point~$k$. 
With only finite data, valid inference is ensured by
constructing \emph{confidence intervals} around an estimator, often by
relying on the estimator's (asymptotic) normality. See
\citet{imbens2004nonparametric} for an overview of standard confidence
interval constructions. 
For example, for the least-squares target in Eq.~\eqref{eq:causal_effect}, a standard estimator is obtained by solving the
empirical version of problem~\eqref{eq:causal_effect}.

We study settings in which the causal graph $G$ is not
given a priori but is learned from $\D$ via causal discovery
algorithms. Denote by $\widehat G$ the graph over $X_1,\dots,X_d$ obtained in a data-driven way. 
Our main technical result can be summarized as follows: whenever we have a way of constructing valid confidence intervals for a causal quantity of interest when the causal graph $G$ is \emph{fixed}, we can adapt the respective method to produce valid confidence intervals when the causal graph $\widehat G$ is \emph{learned from data}. In the following paragraphs we make this statement more precise.


We will denote by $\mathcal{I}_{\widehat G}$ a set of pairs $(i,j)\subseteq [d]\times [d]$ that determines the causal queries of interest. We allow $\mathcal{I}_{\widehat G}$ to depend on the discovered graph $\widehat G$. Therefore, the set of targets is the set
$\{\beta_{\widehat G}^{(i\rightarrow
j)}\}_{(i,j)\in\mathcal{I}_{\widehat G}}$. In the simplest case,
$\mathcal{I}_{\widehat G}$ is a singleton and we are interested in a
single effect. Importantly, $\widehat G$ is random and thus
$\beta_{\widehat G}^{(i\rightarrow j)}$ is a \emph{random inferential
target}.

What makes inferring the effects $\beta_{\widehat G}^{(i\rightarrow j)}$ challenging is the fact that the randomness in the target $\beta_{\widehat G}^{(i\rightarrow j)}$ is coupled with the randomness in the data $\D$ used to perform inference. This issue arises because we use the data twice: once to estimate the causal model $\widehat G$ and another time to perform causal estimation. This double-dipping phenomenon creates a bias: the estimator $\widehat \beta_{\widehat G}^{(i\rightarrow j)}$ can be further from $\beta_{\widehat G}^{(i\rightarrow j)}$ than predicted by classical theory.

To correct this bias, we propose a way to quantify the error of
``naive'' confidence intervals due to double dipping. In particular,
consider a family of confidence intervals
$\mathrm{CI}_{G}^{(i\rightarrow j)}(\alpha; \D)$ that satisfies
\begin{equation}
\label{eq:naive_ints}
\PP{\exists (i,j)\in \mathcal{I}_{G}: \beta_{G}^{(i\rightarrow j)} \not\in \mathrm{CI}_{G}^{(i\rightarrow j)}(\alpha; \D)} \leq \alpha,
\end{equation}
for all $G$ and $\alpha\in(0,1)$. Importantly,
since $G$ is fixed, the target estimand is trivially independent of
the data $\D$. The guarantee \eqref{eq:naive_ints} does \emph{not}
hold when $\widehat G$ is estimated from $\D$.

In the sequel we will show how to make $\mathrm{CI}_{\widehat
G}^{(i\rightarrow j)}(\alpha;\D)$ be \emph{approximately} valid
via randomization, despite the dependence
between $\widehat G$ and $\D$. Specifically, we will compute a corrected error level
$\tilde \alpha$ such that $$\PP{\exists (i,j)\in
\mathcal{I}_{\widehat G}: \beta_{\widehat G}^{(i\rightarrow j)}
\not\in \mathrm{CI}_{\widehat G}^{(i\rightarrow j)}(\tilde \alpha;\D)}
\leq \alpha.$$ Throughout the paper we will use $\mathrm{CI}_{\widehat
G}^{(i\rightarrow j)}(\alpha) \equiv \mathrm{CI}_{\widehat
G}^{(i\rightarrow j)}(\alpha;\D)$ to denote ``standard'' intervals,
which, if $\D$ is independent of $\widehat G$, satisfy the
high-probability guarantee of Eq.~\eqref{eq:naive_ints}.

One simple choice of $\tilde \alpha$ that ensures validity
is obtained via a Bonferroni correction, even if there is no randomization in
the selection. Formally, if $\G$ is the set of all candidate graphs,
then we can write 
\begin{align*}
\PP{\exists (i,j)\in \mathcal{I}_{\widehat G}:
\beta_{\widehat G}^{(i\rightarrow j)} \not\in \mathrm{CI}_{\widehat
G}^{(i\rightarrow j)}(\tilde \alpha;\D)} &\leq \sum_{G\in\mathcal{G}}
\PP{\exists (i,j)\in \mathcal{I}_{ G}: \beta_{ G}^{(i\rightarrow j)}
\not\in \mathrm{CI}_{ G}^{(i\rightarrow j)}(\tilde \alpha;\D)}\\
&\leq
|\mathcal{G}|\tilde\alpha.
\end{align*}
Thus if we set the target miscoverage
probability to be $\tilde \alpha = \frac{\alpha}{|\mathcal{G}|}$, the
miscoverage probability after selection is guaranteed to be at most
$\alpha$. This strategy has a clear drawback of diminishing
statistical power as the number of graphs in $\G$ grows. Our randomization-based proposal can be seen as a more powerful alternative to a Bonferroni correction that likewise comes with distribution-free, finite-sample guarantees. As we will show, our correction yields a choice of $\tilde
\alpha$ independent of the number of candidate graphs.

Another approach to ensuring validity is to perform data splitting: use a fraction of the data for causal discovery and the remaining data for inference. This alternative has the downside of using fewer data points for both graph estimation and inference. We provide a careful theoretical and empirical comparison to data splitting in Section \ref{sec:experimental}. The key takeaway is that our randomization-based approach consistently outperforms data splitting whenever the dimensionality of the data is non-trivial relative to the sample size---which corresponds exactly to the settings in which a naive combination of causal discovery and classical inference leads to inflated type I error.


\subsection{Correcting Inferences via Max-Information}
\label{sec:ada}

We next discuss the key technical tools that we rely on to choose the
corrected level $\tilde \alpha$. 
The basic idea behind our correction is that randomizing the graph selection criterion serves to bound the degree of dependence between the data $\D$
and the learned graph $\widehat G$, which ameliorates the effect of
selection on the validity of subsequent inference. This degree of dependence is formalized via \emph{max-information}.



\begin{definition}[Max-information~\citep{dwork2015generalization}]
 Fix a parameter $\gamma\in(0,1)$. We define the \emph{$\gamma$-approximate
 max-information} between $\D$ and $\widehat G$ as
\begin{equation*}
I^\gamma_\infty (\widehat G; \D) := \max_\O \log \frac{\PP{(\widehat G,\D)\in \O} - \gamma}{\PP{(\widehat G,\tilde\D)\in \O}},
\end{equation*}
where $\tilde\D$ is an i.i.d.\ copy of $\D$ and $\O$ is maximized over all measurable sets.
\end{definition}

A bound on $I^\gamma_\infty (\widehat G; \D)$ provides a way of
bounding the probability of miscoverage when $\widehat G$ is estimated
from $\D$, as long as we can control the same notion of error in
\emph{fixed} graphs $G$. To see this, let $\mathrm{Err}(\alpha)$ denote the
set of graph/data set pairs for which classical intervals miscover:
$\mathrm{Err}(\alpha) = \{(G,\D): \beta_{G}^{(i\rightarrow j)} \not\in
\mathrm{CI}_{G}^{(i\rightarrow j)}(\alpha;\D)\}$. Then, by the
definition of $I^\gamma_\infty (\widehat G; \D)$, we have
\begin{align*}
\label{eq:correction-macro}
\PP{(\widehat G,\D) \in \mathrm{Err}(\alpha)} &\leq \exp\left(I^\gamma_\infty (\widehat G; \D)\right) \PP{(\widehat G,\tilde\D) \in \mathrm{Err}(\alpha)}  + \gamma\\
 &= \exp\left(I^\gamma_\infty (\widehat G; \D)\right) \E\left[\PPst{(\widehat G,\tilde\D)\in \mathrm{Err}(\alpha)}{\widehat G}\right]  + \gamma.
\end{align*}
Since $\tilde\D$ is a fresh sample independent of $\widehat G$, there is no ``double dipping'' and the right-hand side is bounded by $\exp\left(I^\gamma_\infty (\widehat G; \D)\right) \alpha + \gamma$. Thus, if we want coverage at level $1-\alpha^*$ for some $\alpha^*\in (0,1)$, by the previous argument we see that $\mathrm{CI}_{G}^{(i\rightarrow j)}((\alpha^*-\gamma)\exp(-I^\gamma_\infty (\widehat G; \D));\D)$ will have at least $1-\alpha^*$ coverage. In other words, if we aim naively at an error probability equal to
$(\alpha^*-\gamma)\exp\left(-I^\gamma_\infty (\widehat G; \D)\right)$,
then the error probability \emph{after} data-driven graph selection
can be at most $\alpha^*$. Therefore, if we provide a bound on the
approximate max-information between the selected graph $\widehat G$
and the data $\D$, then it suffices to construct intervals
at a more conservative error level to obtain a rigorous finite-sample
correction.

We note that commonly used confidence intervals often have only asymptotic guarantees, meaning $\limsup_n\PP{\beta_{G}^{(i\rightarrow j)} \not\in
\mathrm{CI}_{G}^{(i\rightarrow j)}(\alpha;\D)}\leq \alpha$; our tools and results immediately apply to such intervals as well. Indeed, by the reverse version of Fatou's lemma, we have 
\begin{align*}
\limsup_n \PP{(\widehat G,\D) \in \mathrm{Err}(\alpha)} &\leq \exp\left(I^\gamma_\infty (\widehat G; \D)\right) \E\left[\limsup_n\PPst{(\widehat G,\tilde\D)\in \mathrm{Err}(\alpha)}{\widehat G}\right]  + \gamma\\
&\leq \exp\left(I^\gamma_\infty (\widehat G; \D)\right) \alpha  + \gamma,
\end{align*}
and thus the same argument as above applies.

It remains to understand how to bound the
max-information between $\widehat G$ and $\D$. One approach
studied in the literature on adaptive data analysis is to make the
causal discovery procedure \emph{differentially private}
\citep{dwork2006calibrating}. Roughly speaking, differential privacy
requires that the output of a statistical analysis be randomized in a
way that makes it insensitive to the replacement of a single data
point. In the following, a ``randomized'' algorithm is any algorithm that is allowed to employ a source of randomness independent of the input data in its computations.

\begin{definition}[Differential privacy~\citep{dwork2006calibrating}]
A randomized algorithm $\A$ is $\epsilon$-differentially private for
some $\epsilon\geq0$ if for any two fixed data sets $\D$ and $\D'$ differing in at most one entry and any measurable set $\O$, we
have $\PP{\A(\D)\in \O} \leq e^\epsilon \PP{\A(\D')\in \O}$,
where the probabilities are taken over the randomness of the
algorithm.
\end{definition}

To translate differential privacy into a bound on the max-information,
we apply the following key result due to
\citet{dwork2015generalization}.

\begin{proposition}[\citet{dwork2015generalization}]
\label{prop:max-info}
    Suppose that algorithm $\A$ is $\epsilon$-differentially private, and fix any $\gamma\in(0,1)$. Then, we have
$I^\gamma_\infty(\A(\D);\D) \leq \frac{n}{2}\epsilon^2 + \epsilon \sqrt{n\log(2/\gamma)/2}$.
\end{proposition}

Putting everything together, it suffices to perform causal discovery in a differentially private manner in order to perform valid statistical inference downstream. 
We thus reduce the problem of valid inference after causal discovery to one of developing algorithms for
differentially private causal discovery.

\section{Noisy Causal Discovery}
\label{sec:noisy_selection}


Suppose we have a candidate set $\G$ of causal graphs that captures
our uncertainty about which data-generating model to choose. To select a graph from $\G$, we specify a score function, $S(G,\D)$, which takes as input a graph
$G$ and data set $\D$, and we select the graph with the maximum
score:
\begin{equation}
\label{eqn:score_sel}
\widehat G_* = \argmax_{G\in\G} S(G,\D).
\end{equation}
The score function $S(G,\D)$ is typically formulated as some measure of
compatibility between $G$ and the relationships suggested by the data
$\D$, such as the Bayesian information criterion (BIC). Note that $\widehat G_*$ depends on the data $\D$ and is thus random.

To enable valid statistical inference after graph selection, we rely on a randomized version of the selection rule \eqref{eqn:score_sel}.
The key step is to prove that the randomized selection rule is
differentially private. To accomplish this, one needs to
consider the \emph{sensitivity} of the score.  The
amount of necessary randomization is proportional to the
score sensitivity.

\begin{definition}[Score sensitivity]
\label{def:sensitivity}
    A score function $S(G,\D)$ is \emph{$\tau$-sensitive} if for any graph $G\in\mathcal{G}$ and data sets $\D,\D'$ differing in at most one entry, we have
    $|S(G,\D) - S(G,\D')| \leq \tau$.
\end{definition}

Roughly speaking, score sensitivity bounds the influence that any single data point can have on the choice of the best-scoring graph within the uncertainty set.

\begin{algorithm}[tb]
\SetAlgoLined
\SetKwInOut{Input}{input}
\Input{data set $\D$, set of graphs $\G$, privacy parameter $\epsilon$, $\tau$-sensitive score function $S$}
\textbf{output:} causal graph $\widehat G$\newline
For all $G\in\G$, sample $\xi_G \stackrel{\mathrm{i.i.d.}}{\sim} \mathrm{Lap}\left(\frac{2\tau}{\epsilon}\right)$\newline
Set $\widehat G \leftarrow \argmax_{G\in\G} S(G,\D) + \xi_G$\newline
Return $\widehat G$
\caption{\textsc{noisy-select}}
\label{alg:score-based-sel}
\end{algorithm}

We present our \textsc{noisy-select} method in Algorithm~\ref{alg:score-based-sel}, and
state its privacy guarantee in the following lemma. All proofs can be found in Appendix~\ref{sec:noisy_discovery_proofs}.
\begin{lemma}
\label{lemma:dp_noisy_sel}
The \textsc{noisy-select} algorithm (\Cref{alg:score-based-sel}) is
$\epsilon$-differentially private.
\end{lemma}

Combined with \Cref{prop:max-info}, \Cref{lemma:dp_noisy_sel} implies
a correction in the form of a discounted error level for confidence
interval construction---conceptually similar to a Bonferroni
correction---that ensures valid inference for the causal
effects estimated from~$\widehat G$.

\begin{theorem}
\label{thm:score_fn}
Suppose $\widehat G$ is selected via \textsc{noisy-select} (\Cref{alg:score-based-sel}). Then, for any causal graph
$G\in\mathcal{G}$, we have $$\PP{\exists (i,j)\in\mathcal{I}_{G} :
\beta_{G}^{(i\rightarrow j)} \not\in \mathrm{CI}_{G}^{(i\rightarrow
j)}(\tilde \alpha),~\widehat G = G} \leq \alpha,$$ where $\tilde
\alpha = (\alpha - \gamma) \exp\left(-\frac{n}{2}\epsilon^2 - \epsilon
\sqrt{n\log(2/\gamma)/2}\right)$, for any $\gamma\in(0,\alpha)$.
Consequently, $$\PP{\exists (i,j)\in\mathcal{I}_{\widehat G} :
\beta_{\widehat G}^{(i\rightarrow j)} \not\in \mathrm{CI}_{\widehat
G}^{(i\rightarrow j)}(\tilde \alpha)} \leq \alpha.$$
\end{theorem}


Notably, the correction in Theorem \ref{thm:score_fn} depends only on
$\epsilon$ (essentially, the noise level) and the sample size $n$; it does not
depend on $|\G|$. One principled way to choose $\gamma$ is so as to maximize $\tilde \alpha$, since this minimizes the size of
$\mathrm{CI}_{\widehat G}^{(i\rightarrow j)}(\tilde \alpha)$.

We note that typically the score sensitivity $\tau$ is a decreasing function of $n$, which implies that $\epsilon$ can be chosen as a decreasing function of $n$ in order to keep the noise in Algorithm~\ref{alg:score-based-sel} at a constant level. This in turn allows achieving $\tilde \alpha\rightarrow \alpha$ as $n\rightarrow \infty$ (assuming that $\gamma$ is also tuned so that $\gamma\rightarrow 0$).

We next quantify the suboptimality of the
randomized selection $\widehat G$ relative to the ideal
selection $\widehat G_*$.

\begin{proposition}
\label{prop:noisy_discovery_utility}
Fix $\delta\in(0,1)$. Then, for any graph $G\in\mathcal{G}$ with $S(G,\D) \leq S(\widehat G_*,\D) - \frac{4\tau}{\epsilon}\log(2/\delta)$, \textsc{noisy-select} outputs $G$ with probability at most $\delta$.
\end{proposition}

One immediate consequence of \Cref{prop:noisy_discovery_utility} is
that \Cref{alg:score-based-sel} outputs the optimal graph $\widehat
G_*$ with probability at least $1-\delta$, when $\widehat G_*$ is
``obvious,'' namely when there is no suboptimal graph with score
within $\frac{4\tau}{\epsilon}\log(2|\G|/\delta)$ of $S(\widehat
G_*,\D)$.

Finally, we discuss ways of bounding or evaluating the score sensitivity. One rigorous way to bound the score sensitivity is to make use of \emph{clipping}, thus bounding the contribution of any one data point. To illustrate this point, consider the Bayesian information criterion (BIC), one of the most common scoring criteria in the causal discovery literature. 
When the variables are modeled as Gaussian with variance $\sigma^2$, the BIC is defined as:
\begin{equation*}
S_{\text{BIC}}\left(G,\D\right) = - \min_\theta \frac{1}{n \sigma^2} \sum_{j=1}^d \sum_{k=1}^n \left(X_j^{(k)} - \sum_{s \in \mathbf{Pa}_j^{G}}\theta_s X_s^{(k)} \right)^2 - \sum_{j=1}^d\frac{|\mathbf{Pa}_j^G|}{n} \log n.
\end{equation*}
To achieve bounded
sensitivity, we can simply replace $\left(X_j^{(k)} - \sum_{s \in \mathbf{Pa}_j^{G}}\theta_s X_s^{(k)} \right)^2$ with its clipped version, $\min\left\{\left(X_j^{(k)} - \sum_{s \in \mathbf{Pa}_j^{G}}\theta_s X_s^{(k)} \right)^2,C\right\}$.
It is not hard to see that such a clipped
BIC score is $\frac{Cd}{n\sigma^2}$-sensitive in the worst case. The issue with clipping is that, on one hand, it makes the minimization problem in $S_{\text{BIC}}^C\left(G,\D\right)$ nonconvex and thus computationally intractable, and on the other, the sensitivity in practice might be far smaller than the worst-case bound.

For this reason, we recommend evaluating the sensitivity empirically. There are different ways this could be implemented; for example, one could generate $B$ bootstrap resamples of the data $\D^*_1,\dots,\D^*_B$ and compute $\hat\tau = \max_{G\in\G, i, b} |S(G,\D^*_b) - S(G,\D^{*,-i}_b)|$, where $\D^{*,-i}_b$ is $\D^*_b$ with point $i$ removed. Of course, this is a heuristic approach to evaluating sensitivity, but it allows full flexibility in choosing the score. We will show in our experiments that an empirical bound on the sensitivity does not violate the statistical validity of our proposal.

\section{Noisy Causal Discovery via Greedy Search}
\label{sec:ges}

We extend the randomization scheme in \Cref{alg:score-based-sel}
to causal discovery via \emph{greedy equivalence search}
(GES)~\citep{chickering2002optimal}, which is an efficient alternative to exact search when the latter is
prohibitive computationally.

\subsection{Background on GES}

GES is a procedure that greedily enlarges or reduces the estimated
graph so as to locally maximize a pre-specified score function. The
appeal of GES lies in the fact that it is consistent despite being a greedy search method \citep{chickering2002optimal}.

A core component of the classical GES algorithm is its score function, which
is required to be \textit{decomposable}, meaning that the score of the entire graph can be expressed
as a sum of ``subscores'' obtained by regressing each variable $X_i$
on its parents in $G$.

\begin{definition}[Decomposability]
A scoring criterion $S$ is \emph{decomposable} if there exists a map
$s$ such that, for any DAG $G$ and data set $\D$, we have
$S(G, \D) = \sum_{i=1}^d s(X_i, \mathbf{Pa}_i^G, \D)$.
\end{definition}
\noindent Throughout we will refer to the values $s(X_i, \mathbf{Pa}_i^G,
\D)$ as the \emph{local scores}.

GES greedily enlarges or reduces the selected graph by evaluating the
improvements obtained by either applying an edge insertion or an edge
deletion. Thus, crucial in executing GES are the
\emph{insertion score improvement} and \emph{deletion score
improvement}, respectively:
\begin{align}
\label{eq:score_gain1}
	\Delta S^+(e, G, \D) &\doteq S(G \cup e, \D) - S(G, \D);\\
\label{eq:score_gain2}
	\Delta S^-(e, G, \D) &\doteq S(G \setminus e, \D) - S(G, \D),
\end{align}
where $G\cup e$ denotes the DAG resulting from adding edge $e$ to DAG
$G$ and $G\setminus e$ denotes the DAG resulting from removing edge
$e$ from $G$. Due to decomposability, the score change
implied by adding or removing an edge depends only on the local
structure of $G$ around edge $e$: if $e=X_i\rightarrow X_j$, we have
that $\Delta S^+(e, G, \D) = s(X_j,\mathbf{Pa}_j^{G}\cup X_i, \D) - s(X_j,\mathbf{Pa}_j^{G}, \D)$.
A similar identity holds for $\Delta S^-(e, G, \D)$.

Given a decomposable score, the classical GES algorithm works as
follows. Throughout the execution, GES maintains a CPDAG $\widehat G$.
In the first half of the execution, in each sequential round GES
considers all CPDAGs that could be obtained by applying a valid edge
insertion operator, which we refer to as ``$(+)$-operators,'' to $\widehat G$. For
all possible $(+)$-operators $e$, GES evaluates the score gain,
$\Delta S^+(e, \widehat G, \D)$. Note that we slightly
abuse notation since $\widehat G$ is a CPDAG and not a single DAG
and $e$ includes specifications in addition to an edge.  More formally,
the score gain of an insertion operator $\Delta S^+(e,
\widehat G,
\D)$ is computed as in \eqref{eq:score_gain1} for a specific DAG $G$
consistent with the CPDAG $\widehat G$ (see Corollary 16 in
\citet{chickering2002optimal} for details). Once all possible edge
insertions have been scored, GES finds the $(+)$-operator $e^*$ that
maximizes the gain, $e^* = \argmax_e\Delta S^+(e, \widehat G, \D)$. If
$\Delta S^+(e^*, G, \D) > 0$---meaning that applying operator $e^*$
improves upon the score of the current graph---the algorithm applies
$e^*$ to $\widehat G$ and repeats the same insertion operator
selection procedure. Otherwise, if a local maximum is reached, it
halts. After the local maximum is reached, GES performs an analogous
sequence of steps once again, only now considering edge removal
operators, which we refer to as ``$(-)$-operators,'' and the corresponding
score gains $\Delta S^-(e,\widehat G,\D)$. As before, this score gain
is evaluated for a specific DAG consistent with $\widehat G$ according
to Eq.~\eqref{eq:score_gain2}
\citep[see Corollary 18 in][for details]{chickering2002optimal}.

The randomization scheme of \textsc{noisy-ges} is agnostic to certain graph-theoretic aspects of GES, including what
constitutes a valid edge insertion or edge removal operator for a CPDAG
and whether GES maintains a single DAG or a CPDAG. These choices likewise do not affect the implied max-information bound. For this reason, we skip these details in the main text and review them in Appendix \ref{app:ges_appendix}. The irrelevance of these details also implies that
one can view GES intuitively as operating on the space of
DAGs, rather than CPDAGs, and $(+)$-operators (resp. $(-)$-operators) as being single-edge additions (resp. removals) that maintain the DAG structure.

\subsection{Noisy GES}

To enable valid statistical inference after causal discovery via GES,
we develop a differentially private variant of GES that relies on randomization. The GES algorithm utilizes the data in two basic
ways: by selecting the best-scoring operator and by checking whether
applying the corresponding operator leads to a score improvement.
Hence, in order to make GES differentially private, we compute
noisy scores and apply a randomized rule for stopping at a local
maximum. We use the Report Noisy Max
mechanism and the Above Threshold mechanism
\citep{dwork2014algorithmic} for the two objectives, respectively.


Similarly to the case of exact search, we require that the local scores have low sensitivity.

\begin{definition}[Local score sensitivity]
A local score function $s$ is \emph{$\tau$-sensitive} if for all indices
$i\in[d], I\subseteq [d]$ and data sets $\D,\D'$ differing in at most one entry, we have $|s(X_i,X_I,\D) -
s(X_i,X_I,\D')| \leq \tau$.
\end{definition}

Note that local score sensitivity immediately implies a bound on the
sensitivity of $\Delta S^{\texttt{sgn}}$, for $\texttt{sgn} \in\{+,-\}$.
If $s$ is $\tau$-sensitive, we have $|\Delta S^{\texttt{sgn}}(e, G, \D) - \Delta S^{\texttt{sgn}}(e, G, \D')| \leq 2\tau$,
for $\texttt{sgn} \in\{+,-\}$. This bound holds for all edges $e$ and
graphs $G$. As in Section \ref{sec:noisy_selection}, we can evaluate the local score sensitivity empirically or deterministically bound it by clipping.

We state the \textsc{noisy-ges}
algorithm along with its privacy guarantees. We stress that \textsc{noisy-ges} is equally valid for greedy search over CPDAGs and
greedy search over DAGs.

\begin{algorithm}[t]
\SetAlgoLined
\SetKwInOut{Input}{input}
\Input{data set $\D$, maximum number of edges $E_{\max}$, score $S$ with local score sensitivity $\tau$, privacy parameters $\epsilon_{score}, \epsilon_{\mathrm{thresh}}$}
\textbf{output:} causal graph $\widehat G$\newline
Initialize $\widehat G$ to be an empty graph\newline
Run forward pass $\widehat G \leftarrow \text{GreedyPass}(\widehat G,\D,E_{\max},S,\tau,\epsilon_{\mathrm{score}},\epsilon_{\mathrm{thresh}},+)$\newline
Run backward pass $\widehat G \leftarrow \text{GreedyPass}(\widehat G, \D,E_{\max},S, \tau,\epsilon_{\mathrm{score}},\epsilon_{\mathrm{thresh}},-)$\newline
Return $\widehat G$
\caption{\textsc{noisy-ges}}
\label{alg:GES_general}
\end{algorithm}

\begin{algorithm}[ht]
\SetAlgoLined
\SetKwInOut{Input}{input}
\Input{initial graph $\widehat G_0$, data set $\D$, maximum number of edges $E_{\max}$, score $S$ with local score sensitivity~$\tau$, privacy parameters $\epsilon_{\mathrm{score}}, \epsilon_{\mathrm{thresh}}$, pass indicator $\texttt{sgn}\in\{+,-\}$}
\textbf{output:} estimated causal graph $\widehat G$\newline
Initialize $\widehat G\leftarrow \widehat G_0$\newline
Sample noisy threshold $\nu \sim \text{Lap}\left(\frac{4 \tau}{\epsilon_{\mathrm{thresh}}}\right)$\newline
\For{$t=1,2,\dots,E_{\max}$}{
\ Construct set $\mathcal{E}^{\texttt{sgn}}_t$ of valid (\texttt{sgn})-operators\newline
For all $e\in \mathcal{E}^{\texttt{sgn}}_t$, compute $\Delta S^{\texttt{sgn}}(e, \widehat{G}, \D)$ and sample $\xi_{t,e} \stackrel{\text{i.i.d.}}{\sim} \text{Lap}\left(\frac{4 \tau}{\epsilon_{\mathrm{score}}}\right)$\newline
Set $e^*_t = \argmax_{e\in \mathcal{E}^{\texttt{sgn}}_t} \Delta S^{\texttt{sgn}}(e, \widehat{G}, \D) + \xi_{t,e}$\newline
Sample $\eta_t \sim \text{Lap}\left(\frac{8 \tau}{\epsilon_{\mathrm{thresh}}}\right)$\newline
\uIf{$\Delta S^{\texttt{sgn}}(e^*_t, \widehat G, \D) + \eta_t > \nu$}
{
Apply operator $e_t^*$ to $\widehat G$
}
\Else{
break
}
}
Return $\widehat G$
\caption{GreedyPass}
\label{alg:GES_single_pass_general}
\end{algorithm}

\begin{lemma} 
\label{lemma:GES_dp_guarantee}
The \textsc{noisy-ges} algorithm (\Cref{alg:GES_general}) is $(2\epsilon_{\mathrm{thresh}} + 2E_{\max}\epsilon_{\mathrm{score}})$-differentially private.
\end{lemma}

With Lemma \ref{lemma:GES_dp_guarantee} in hand, we can ensure valid inference
after causal discovery. We state an analogue of
\Cref{thm:score_fn} for \textsc{noisy-ges} which shows how to discount
the target miscoverage level in order to preserve validity after graph
discovery via greedy search. The result follows by putting together
\Cref{prop:max-info} and \Cref{lemma:GES_dp_guarantee}.

\begin{theorem} 
\label{thm:ges_inf}
Suppose that we select $\widehat G$ via noisy greedy equivalence search (\Cref{alg:GES_general}). Then, for any causal graph $G$, we have
$$\PP{\exists (i,j) \in \mathcal{I}_{G}: \beta^{(i\rightarrow j)}_{G} \not\in \mathrm{CI}_{G}^{(i\rightarrow j)}(\tilde \alpha),~\widehat G = G} \leq \alpha,$$
where
$\tilde \alpha = (\alpha - \gamma) \exp\left(-2n(\epsilon_{\mathrm{thresh}} + E_{\max}\epsilon_{\mathrm{score}})^2 - (\epsilon_{\mathrm{thresh}} + E_{\max}\epsilon_{\mathrm{score}})\sqrt{2n\log(1/\gamma)}\right)$,
for any $\gamma\in(0,\alpha)$. Consequently,
$$\PP{\exists (i,j) \in \mathcal{I}_{\widehat G}: \beta^{(i\rightarrow j)}_{\widehat G} \not\in \mathrm{CI}_{\widehat G}^{(i\rightarrow j)}(\tilde \alpha)} \leq \alpha.$$
\end{theorem}

Notice that setting $\epsilon_{\mathrm{score}},\epsilon_{\mathrm{thresh}} \propto \frac{1}{\sqrt{n}}$ (and $\gamma$ to be a small constant fraction of $\alpha$) implies an essentially constant discount factor, i.e., a constant ratio between $\alpha$ and $\tilde \alpha$. This will be our default scaling for $\epsilon_{\mathrm{score}}$ and $\epsilon_{\mathrm{thresh}}$.

\section{Empirical Studies}
\label{sec:experimental}

We complement our theoretical findings with experiments.
First, in Section \ref{sec:exp_validity} we evaluate the severity of the error incurred by uncorrected inference after causal discovery and compare to the error of our methods, \textsc{noisy-select} and \textsc{noisy-ges}. We find that the double-dipping phenomenon that motivates our work indeed leads to invalid inference, worsening in low-sample and high-dimensional regimes, and that our methods effectively solve this. Then, in Section \ref{sec:exp_quality} we compare the quality of the graph discovered by our randomized methods with that of the graph found by standard causal discovery with data splitting. We apply our methods with a robustified version of the BIC score that we informally refer to as the ``Huber score.'' It simply replaces the squared loss with the Huber loss for increased robustness. For GES we use the usual BIC score. In Appendix \ref{subsec:huber_validity} we show that our findings are qualitatively the same if we apply GES with the robustified score.

\subsection{Validity}\label{sec:exp_validity}

We quantify the severity of the error incurred by uncorrected inference after causal discovery by evaluating the probability of miscoverage of a causal estimand. In particular, we use the same data both to estimate the causal graph $\widehat G$ either via exact selection of the score-maximizing graph or GES \emph{and} to compute infer the effect $\widehat \beta_{\widehat G}^{(i\rightarrow j)}$. For exact selection, we take the graph with the maximum score among the true one and nine additional variations created by first removing an existing edge with probability $p_\mathrm{remove}$ and then adding each possible new edge $i\rightarrow j$ with probability $p_\mathrm{add}$. We select the causal target $i\rightarrow j$ with uniform probability over all edges in $\widehat G$ and use a standard z-interval to produce a $95\%$ confidence interval for the effect $\beta_{\widehat G}^{(i\rightarrow j)}$. We investigate two models for generating the true underlying graph.

\paragraph{Empty graph.} In the first model, our goal is to show that pure noise can be misconstrued into an effect without proper correction.
We draw $n$ independent samples from a $d$-dimensional standard Gaussian distribution, $\mathcal{N}(0, I)$; this corresponds to an empty graph being the true underlying model.
We repeat the experiment $100$ times to estimate the probability of miscoverage of the population-level parameter $\beta^{(i\rightarrow j)}_{\widehat G}$, which in this case is simply zero. This probability is equivalent to the probability of falsely rejecting the null hypothesis that there is no effect between $X_i$ and $X_j$ at significance level $95\%$.
If the estimated graph is empty, no interval is computed and we automatically count those trials as correctly covering the target. In Figure  \ref{fig:validity_empty_graph_select} we plot the probability of error for varying sample size $n$ and number of variables~$d$ for exact selection and \textsc{noisy-select} between ten graph variants created by setting $p_\mathrm{remove}=0$ and $p_\mathrm{add} = 0.01$. In Figure \ref{fig:validity_plots_empty} we show a similar comparison for classical GES and \textsc{noisy-ges}. For both exact selection and classical GES, we see that the error probability exceeds the target error across the board, the violation getting worse as $n$ decreases and as $d$ increases in the case of GES. For our method, we see that the error is controlled at the nominal level. We see that the higher setting of $\epsilon$ appears to be more conservative. This may be because for higher $\epsilon$ the algorithm is more likely to correctly identify the empty graph, in which case no error is made. In Appendix~\ref{subsec:huber_validity} we show the behavior of classical GES when using the Huber score and observe similar behavior to the one in Figure \ref{fig:validity_plots_empty}, showing that the validity is not improved simply by using a robust score.

\begin{figure}[t]
     \centering
     \includegraphics[width=0.32\textwidth]{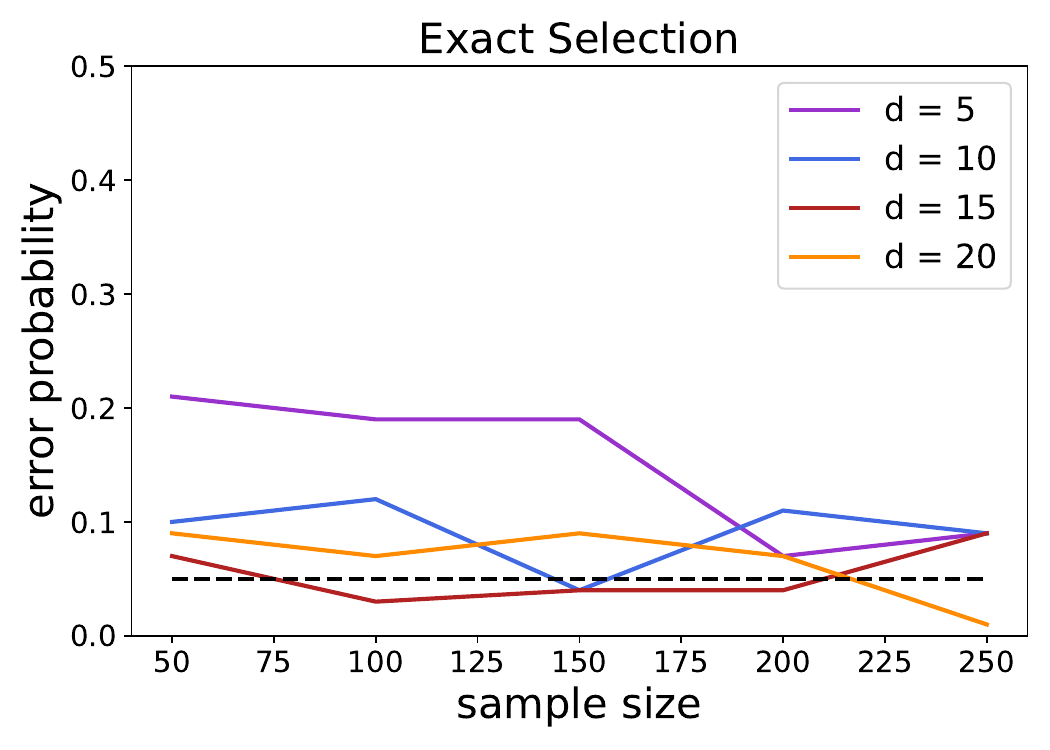}
     \includegraphics[width=0.32\textwidth]{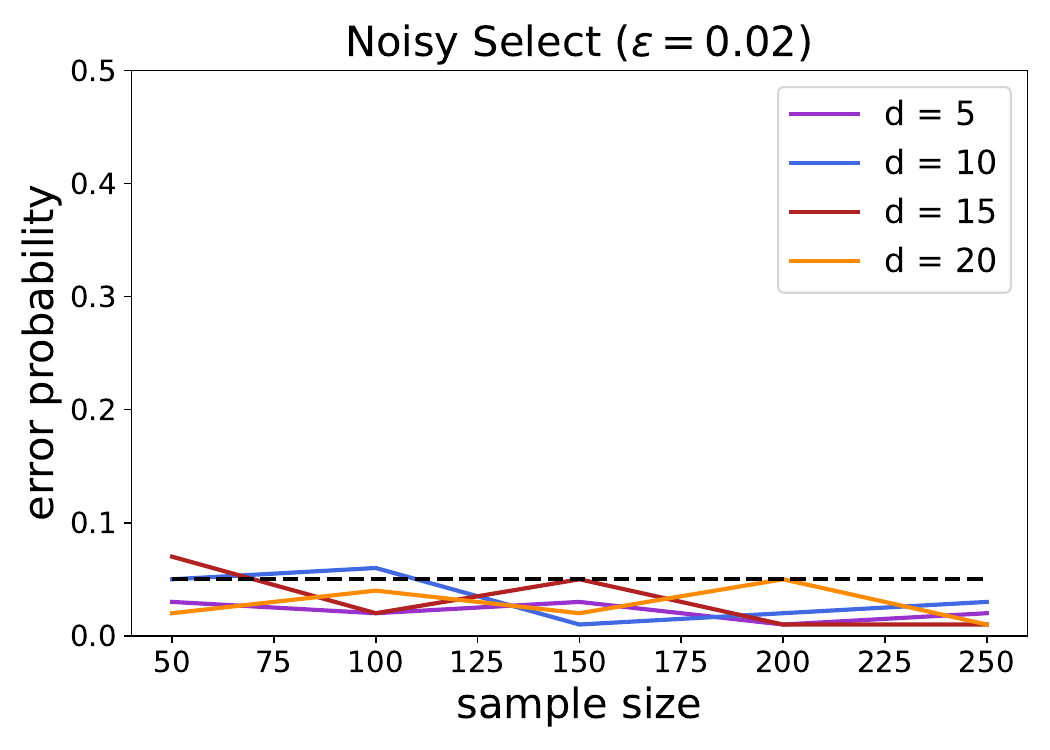}
     \includegraphics[width=0.32\textwidth]{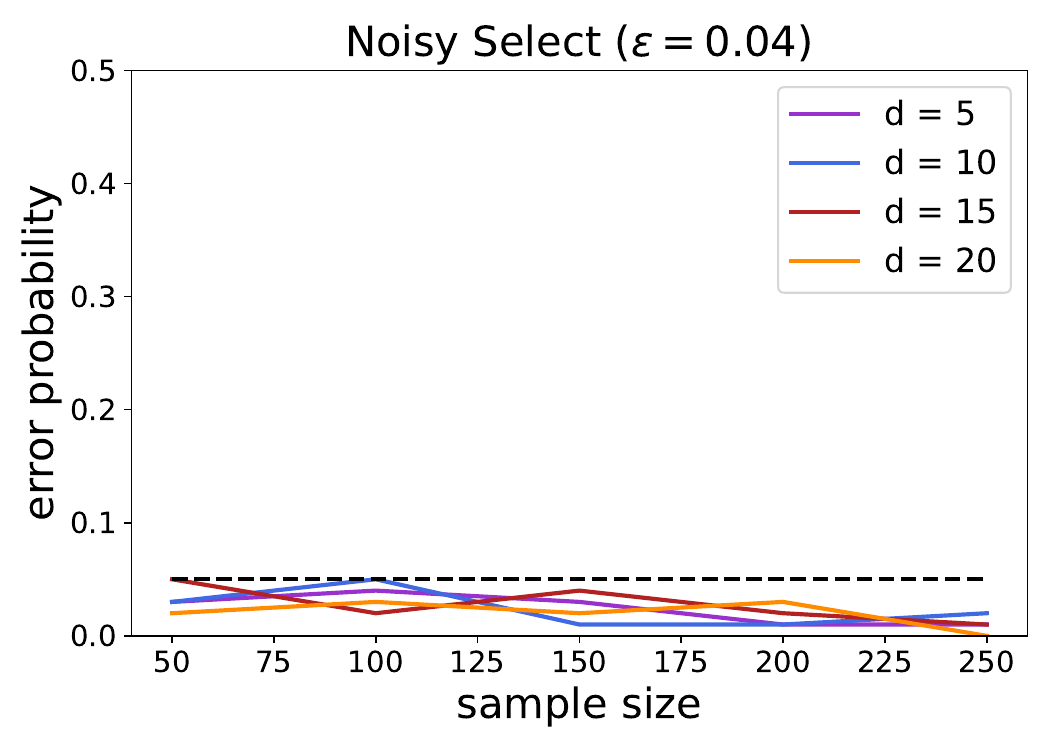}
    \caption{Probability of error for varying $n$ and $d$ in empty graph for exact selection (left), \textsc{noisy-select} with $\epsilon=0.02$ (middle), and $\epsilon=0.04$ (right). }
        \label{fig:validity_empty_graph_select}
\end{figure}

\begin{figure}[t]
     \centering
     \includegraphics[width=0.32\textwidth]{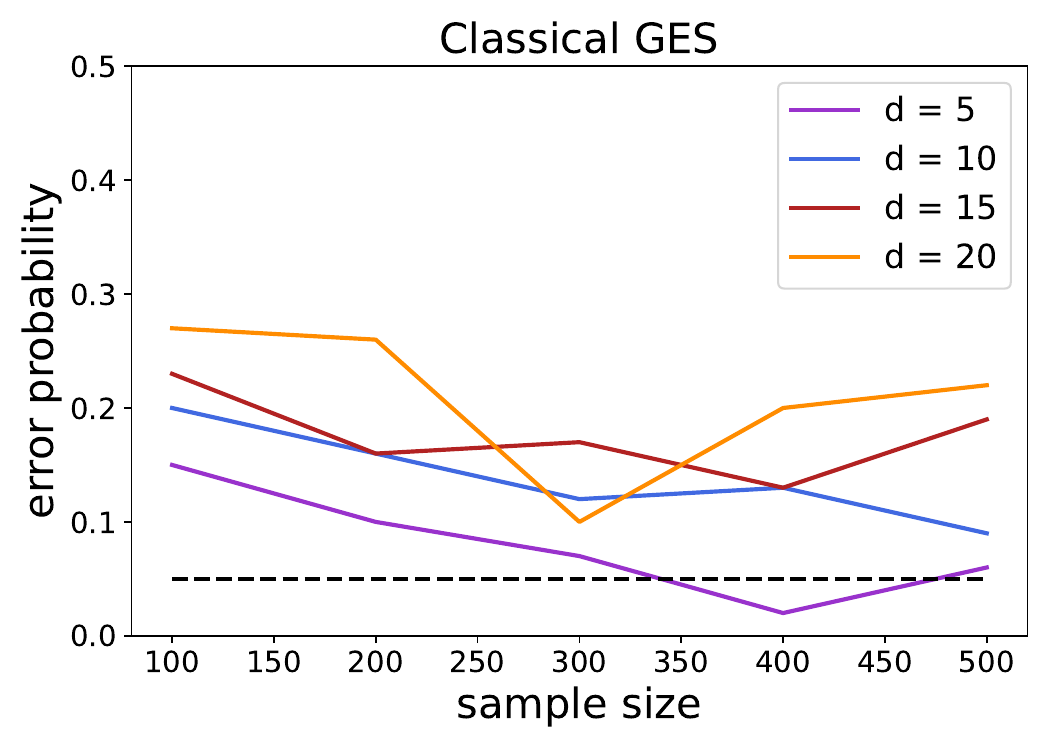}
    \includegraphics[width=0.32\textwidth]{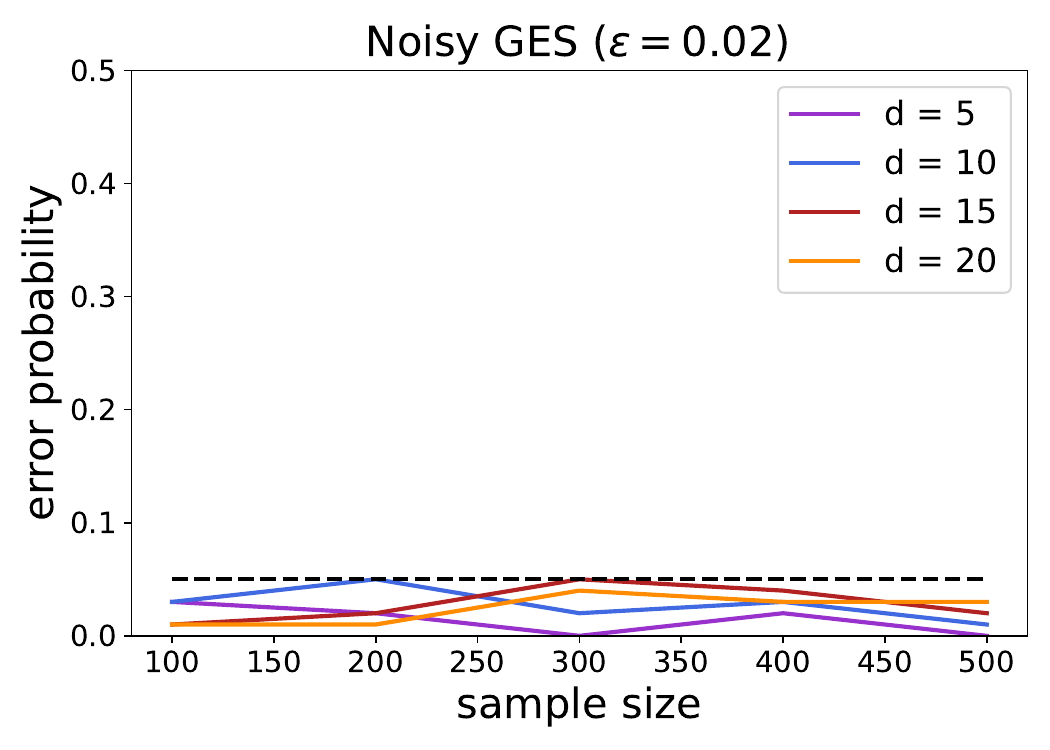}
    \includegraphics[width=0.32\textwidth]{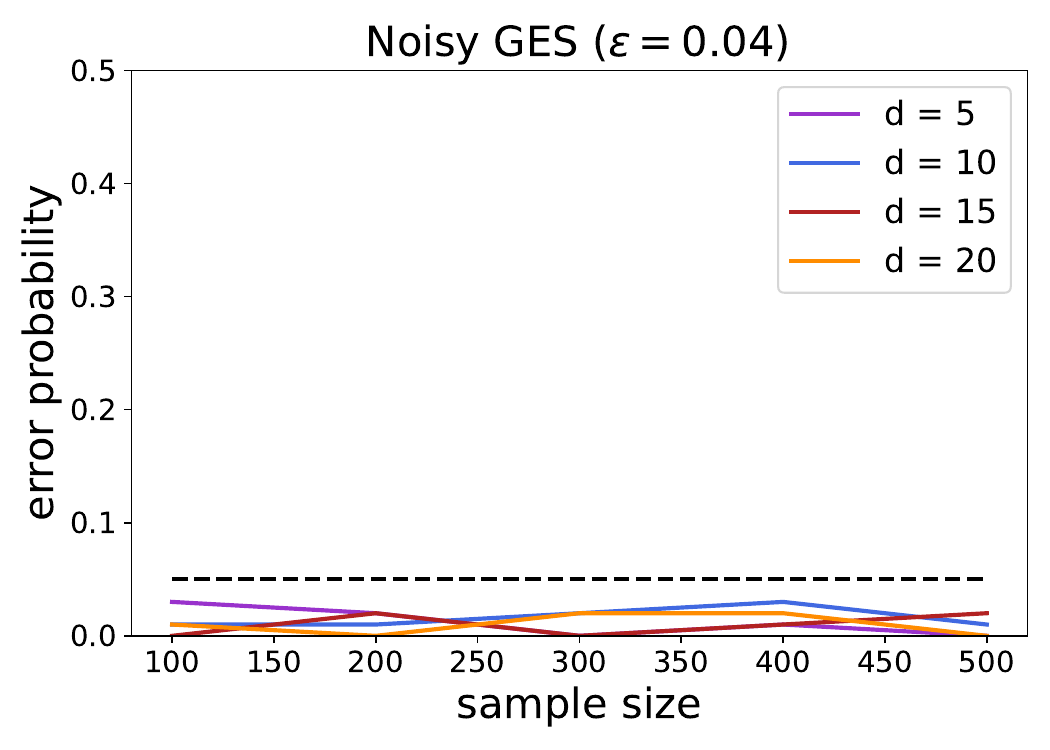}
        \caption{Probability of error for varying $n$ and $d$ in empty graph for classical GES (left), \textsc{noisy-ges} with $\epsilon=0.02$ (middle), and $\epsilon=0.04$ (right).}
        \label{fig:validity_plots_empty}
\end{figure}

\paragraph{Sparse random graph.} 
We next consider a more challenging setting where the underlying DAG is sparse, but there exist truly significant relationships between variables. Formally, we generate an Erd\H{o}s-R\'enyi graph with $d$ nodes and average degree 1, and orient the edges according to a random ordering (while preserving the DAG structure), creating a connectivity matrix $\mathbf{W}$. The entries of $\mathbf{W}$ are either zero or a value sampled uniformly between 2 and 4, and we draw samples as:
$$X^{(k)} \stackrel{\mathrm{i.i.d.}}{\sim} \mathcal{N}\left(\mathbf{0}, ((\mathbf{I}_d - \mathbf{W})(\mathbf{I}_d - \mathbf{W}^\top))^{-1}\right). $$
In this case, the target regression coefficients \eqref{eq:causal_effect} are no longer zero.
We provide analogous comparisons to Figure \ref{fig:validity_empty_graph_select} and \ref{fig:validity_plots_empty}  in Figure~\ref{fig:validity_random_graph_select} and Figure~\ref{fig:validity_plots_random}. Note that this time we evaluate validity with respect to the projected effect $\beta_{\widehat{G}}$ (simulated using $10^6$ new data points). Since the graph is non-empty, we set $p_\mathrm{remove} = 2.5/|E|$ and $p_\mathrm{add} = 0.01$. For both exact selection and classical GES, we observe a qualitatively similar trend as in the empty graph setting, albeit to a less extreme extent. For \textsc{noisy-select} and \textsc{noisy-ges}, we see that for a non-negligible sample size the error is controlled at the desired level. Again, in Appendix~\ref{subsec:huber_validity} we show that classical GES with the Huber score similarly violates the coverage requirement.

\begin{figure}[h]
     \centering
     \includegraphics[width=0.32\textwidth]{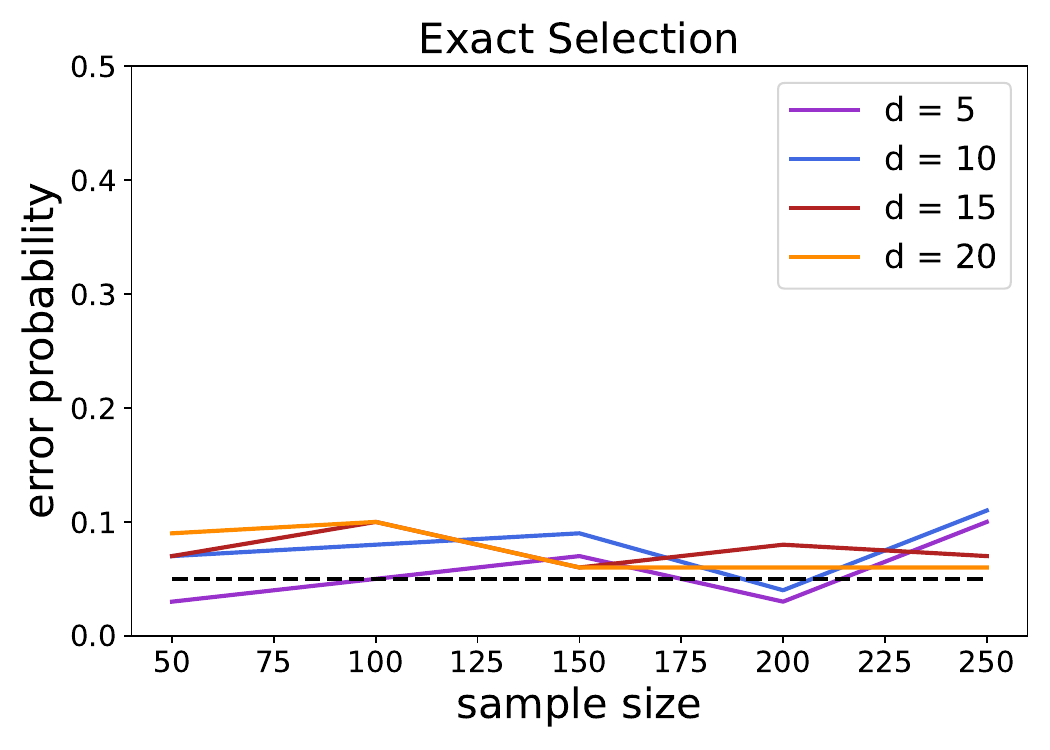}
     \includegraphics[width=0.32\textwidth]{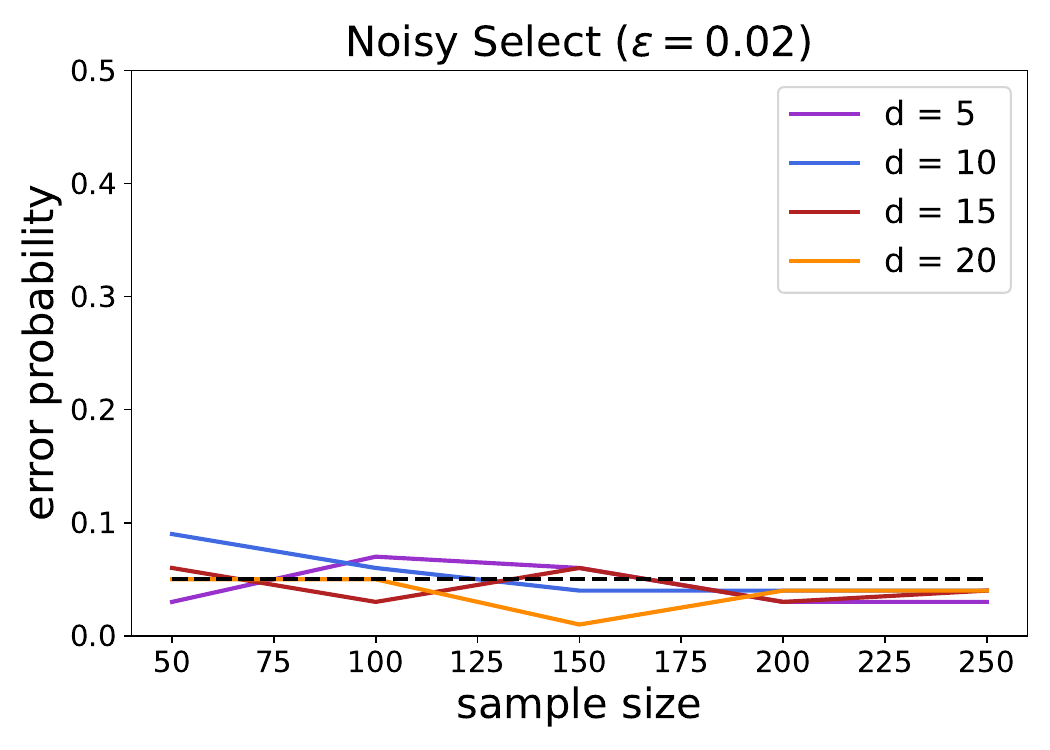}
     \includegraphics[width=0.32\textwidth]{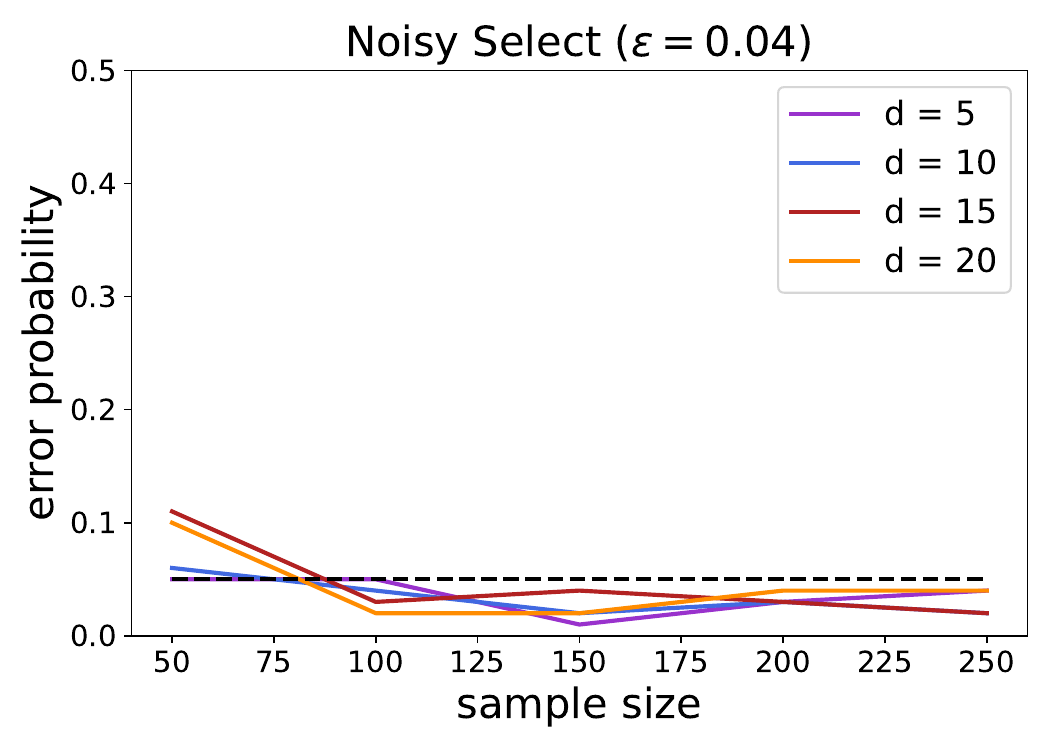}

    \caption{Probability of error for varying $n$ and $d$ in random graph for exact selection (left), \textsc{noisy-select} with $\epsilon=0.02$ (middle), and $\epsilon=0.04$ (right). }
        \label{fig:validity_random_graph_select}
\end{figure}

\begin{figure}[h]
     \centering
     \includegraphics[width=0.32\textwidth]{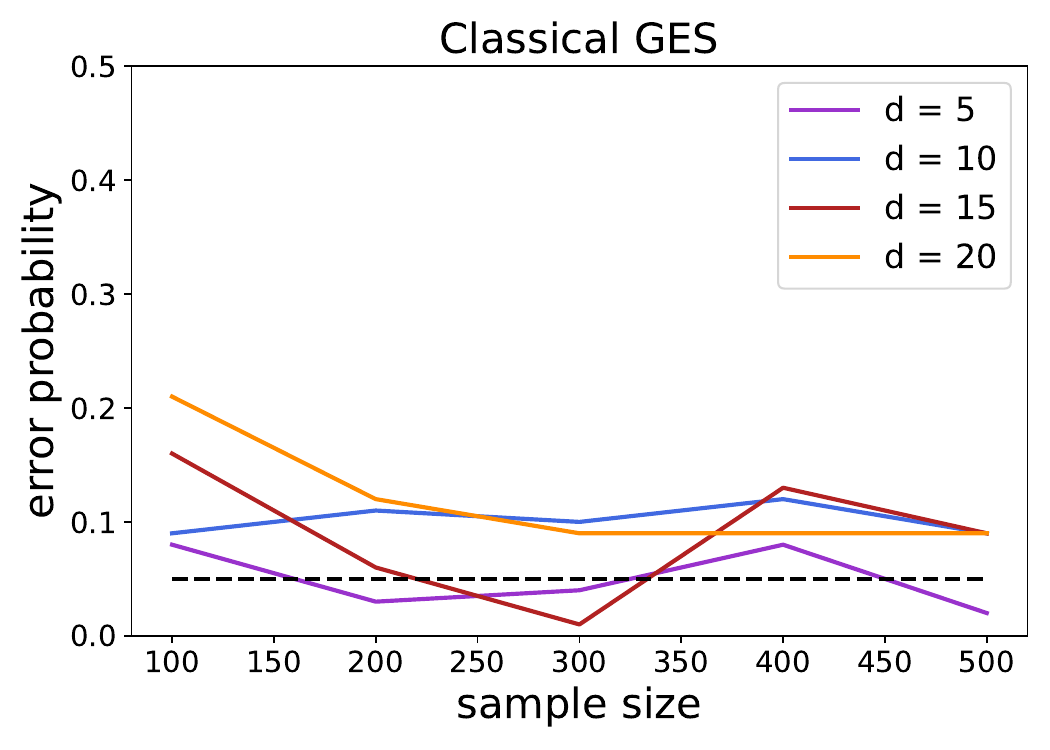}
     \includegraphics[width=0.32\textwidth]{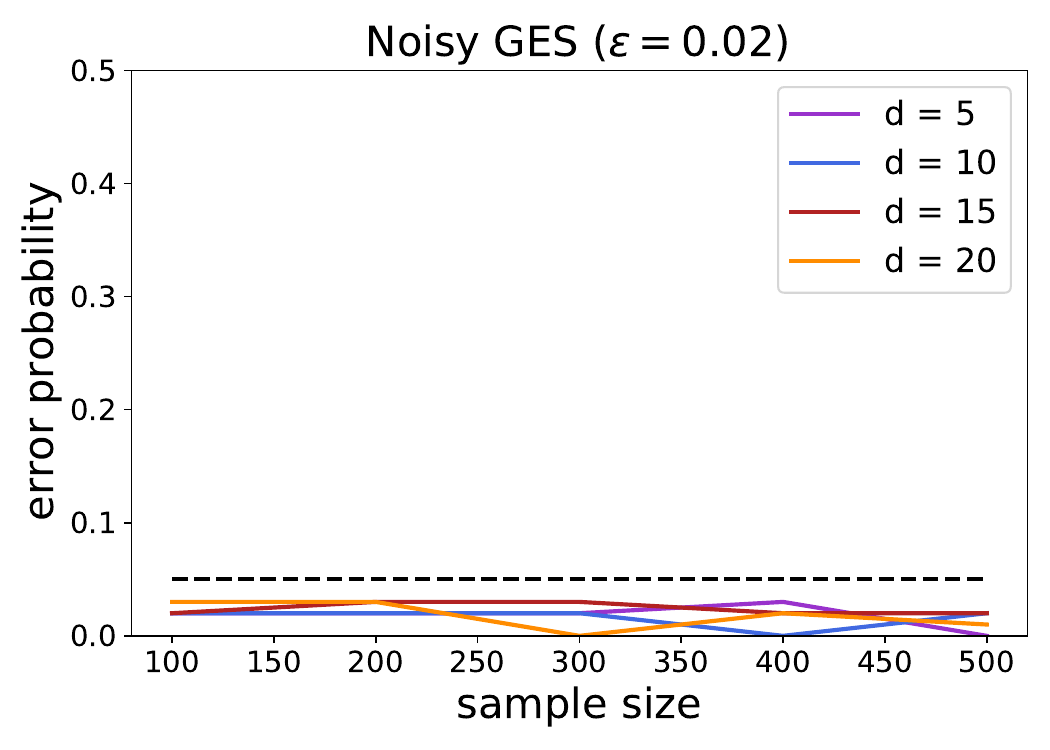}
     \includegraphics[width=0.32\textwidth]{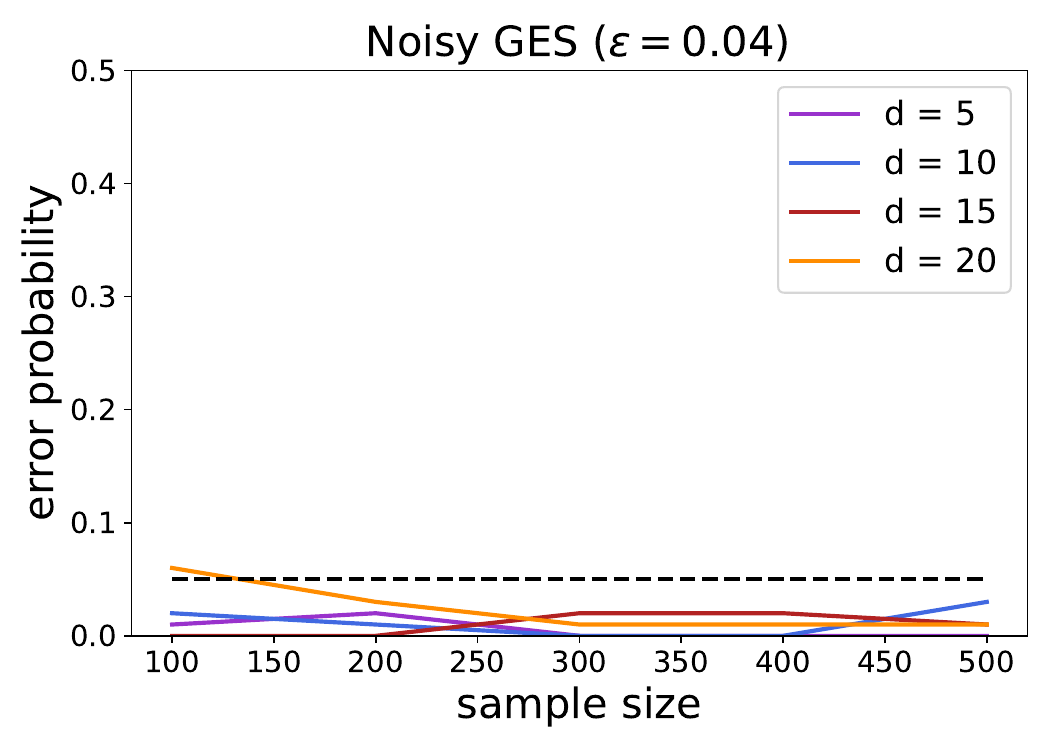}

    \caption{Probability of error for varying $n$ and $d$ in random graph for classical GES (left) and \textsc{noisy-ges} with $\epsilon=0.02$ (middle), and $\epsilon=0.04$ (right).}
        \label{fig:validity_plots_random}
\end{figure}

\subsection{Graph Quality}
\label{sec:exp_quality}

We now compare our randomized methods---\textsc{noisy-select} and \textsc{noisy-ges}---with standard, noiseless selection strategies combined with data splitting. Recall that both approaches provide a valid correction for inference after causal discovery. Our randomized methods use the \emph{whole} data set to learn a graph and perform subsequent inference, but they inject noise into the causal discovery step and use a more conservative error level for inference. Data splitting, on the other hand, splits the data into two independent chunks, estimating the causal graph on one and performing inference on the other. 

More thoroughly, in data splitting we use, up to rounding error, $(1-p)n$ data points to learn the graph $\widehat G$ and the remaining $pn$ data points to do inference, for some splitting fraction $p\in(0,1)$. The parameter $p$ interpolates between two extremes, one in which all data is used for causal discovery ($p=0$) and the other in which all data is reserved for inference ($p=1$). In our framework, the privacy parameters interpolate between these two extremes in a similar fashion (when the max-information is infinite and equal to zero, respectively).

In the following subsections, we compare the two approaches across two metrics: (i) structural Hamming distance (SHD) to the graph found by the classical, noiseless method on the \emph{whole} data, and (ii) confidence interval width. We vary the privacy parameter and compare to data splitting with a low, medium, and high splitting fraction, $p\in\{0.05, 0.5, 0.95\}$.

\subsubsection{Exact Search}

We evaluate the performance of \textsc{noisy-select} (Algorithm~\ref{alg:score-based-sel}) on both synthetic and real data. We note that in the real-data experiments the ground-truth graph is provided by the data curator and may not be perfectly ``true.'' For each data set, we create nine additional variations of the ground-truth graph by first removing each existing edge with probability $p_\mathrm{remove}$ and then adding each possible new edge $i \rightarrow j$ with probability $p_\mathrm{add}$. We then score the ten variations and choose the best one, either using \textsc{noisy-select} or exact, noiseless selection with data splitting. We report the average SHD to the exact maximum using the \emph{whole} data. Below we give the specifics of the experimental setups.

\paragraph{Synthetic data.} First we consider a random graph setting. We fix $d=15$ and $n=100$ and take $p_\mathrm{remove} = 2.5/|E|$, $p_\mathrm{add} = 0.01$. We vary $\epsilon$ between $0.001$ and $0.32$ and plot the SHD to the graph found by GES on the full data set and the interval widths. We plot the average results over $100$ trials in Figure~\ref{fig:stable_select_synthetic}. We observe that \textsc{noisy-select} strongly outperforms the lowest split fraction, while having comparable interval widths for a large fraction of the $\epsilon$ settings.

\begin{figure}[t]
     \centering
     \includegraphics[width=0.4\textwidth]{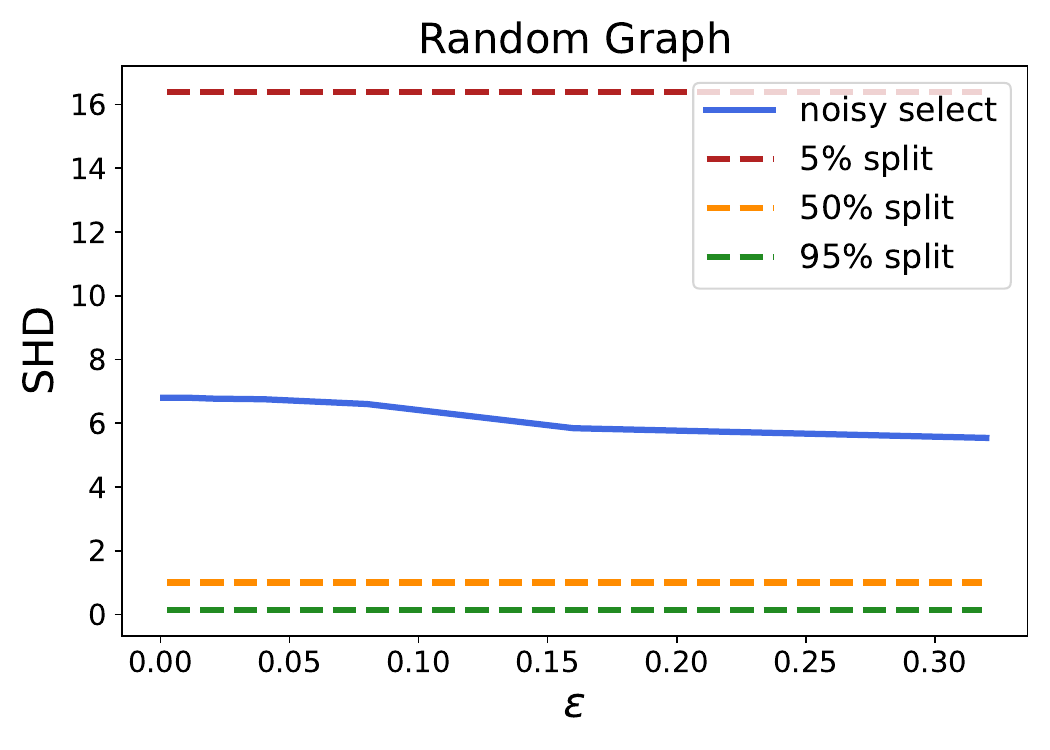}\hspace{0.5cm}
    \includegraphics[width=0.4\textwidth]{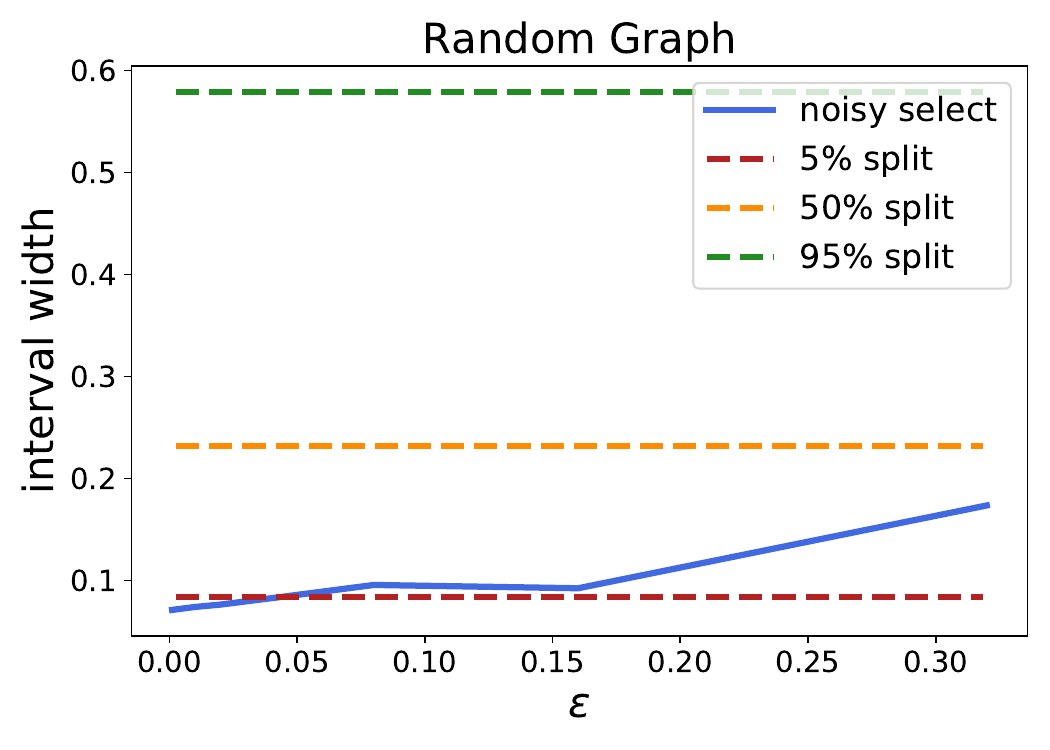}
    \caption{Comparison of \textsc{noisy-select} with varying $\epsilon$ and three data splitting baselines in terms of SHD (left) and interval widths (right) for a random graph.}
    \label{fig:stable_select_synthetic}
\end{figure}

\paragraph{Auto-MPG data.} For our first real-data experiment, we use the Auto-MPG data \citep{quinlan1993combining} contained in the Tuebingen database \citep{mooij2016distinguishing}. According to \cite{quinlan1993combining}, ``the data concerns city-cycle fuel consumption in miles per gallon.'' The data consists of $d=5$ variables (displacement, MPG, horsepower, weight and acceleration) and we take $n=100$ for our experiment. We normalize the data and set $p_\mathrm{remove} = 2/|E|, p_\mathrm{add} = 2/d^2$. We vary $\epsilon$ between $0.01$ and $0.1$ and plot averages over $100$ trials in Figure~\ref{fig:auto_mpg}. We observe that \textsc{noisy-select} interpolates between the $5\%$ and $50\%$ splitting baselines both in terms of SHD and interval widths.

\begin{figure}[b]
\centering
\includegraphics[width=0.4\textwidth]{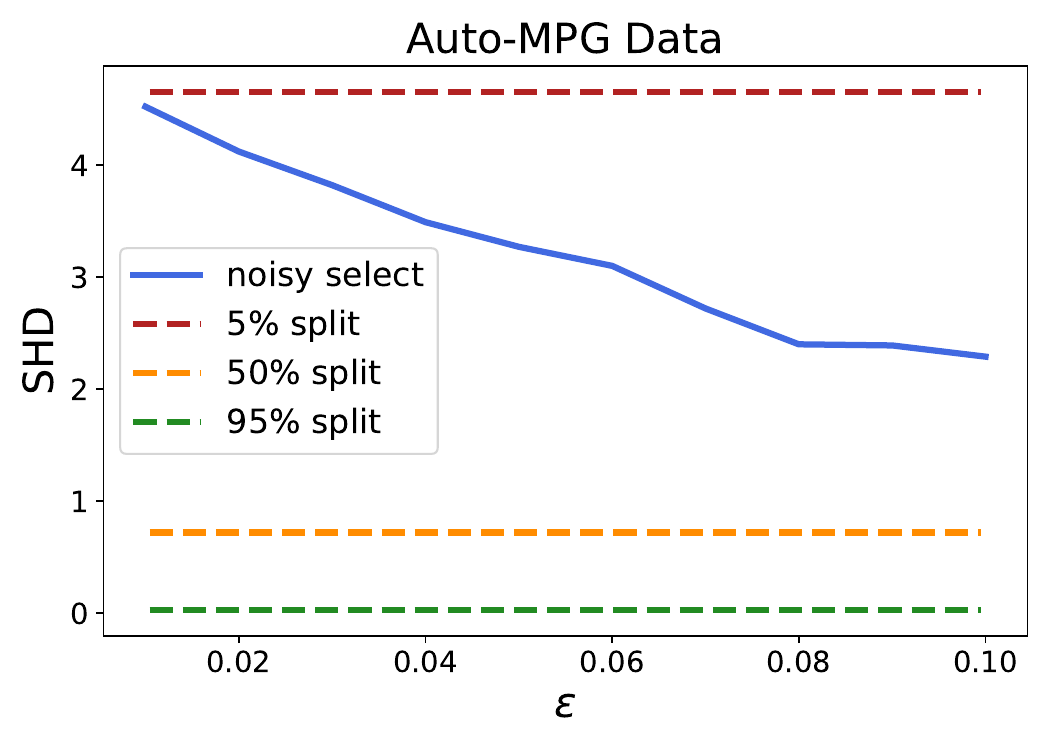}
\hspace{0.5cm}
\includegraphics[width=0.4\textwidth]{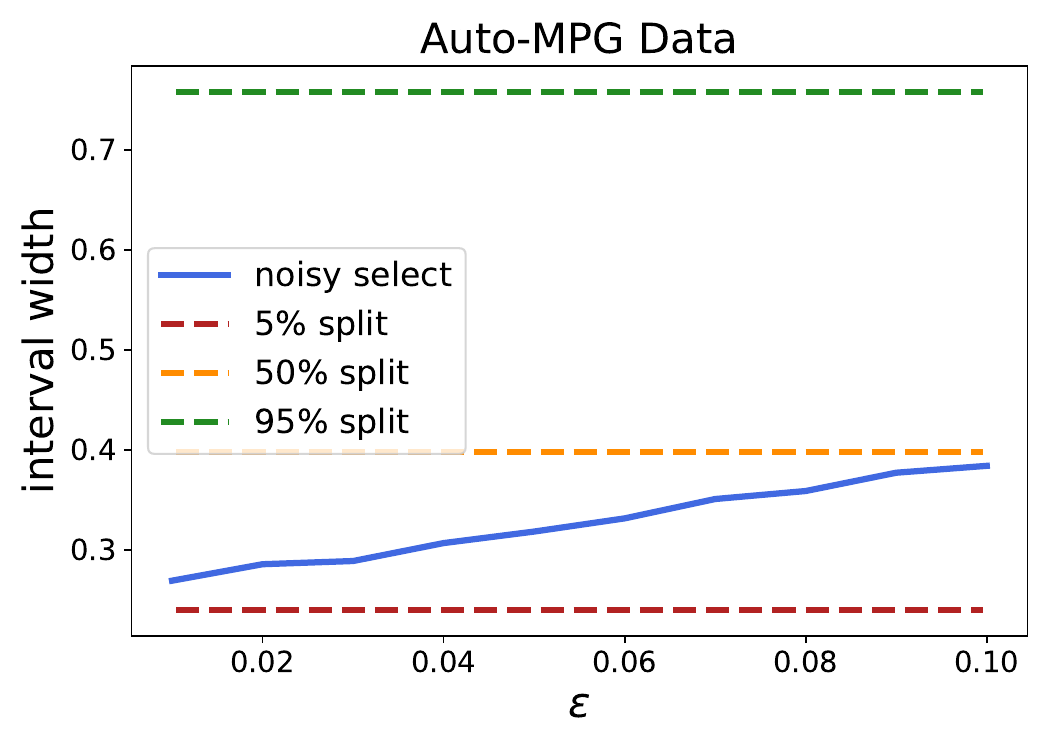}
\caption{Comparison of \textsc{noisy-select} with varying $\epsilon$ and three data splitting baselines in terms of SHD (left) and interval widths (right) on the Auto-MPG data \citep{quinlan1993combining, mooij2016distinguishing}.}\label{fig:auto_mpg}\end{figure}

\paragraph{Flow cytometry data.} For our second real-data experiment, we use flow cytometry data from \cite{sachs2005causal}. The data set has $d=11$ variables and we again take $n=100$ for our experiment. We normalize the data to have unit variance. We take $p_\mathrm{remove} = 5/|E|, p_\mathrm{add} = 5/d^2$. We vary $\epsilon$ between $0.005$ and $0.08$ and plot average structural Hamming distance and resulting average confidence interval widths over $100$ trials in Figure~\ref{fig:sachs_huber}. We observe that \textsc{noisy-select} interpolates between the interval widths of $5\%$ and $50\%$ splits while also achieving relative SHD smaller than using the smaller split.

\begin{figure}[t]
     \centering
         \includegraphics[width=0.4\textwidth]{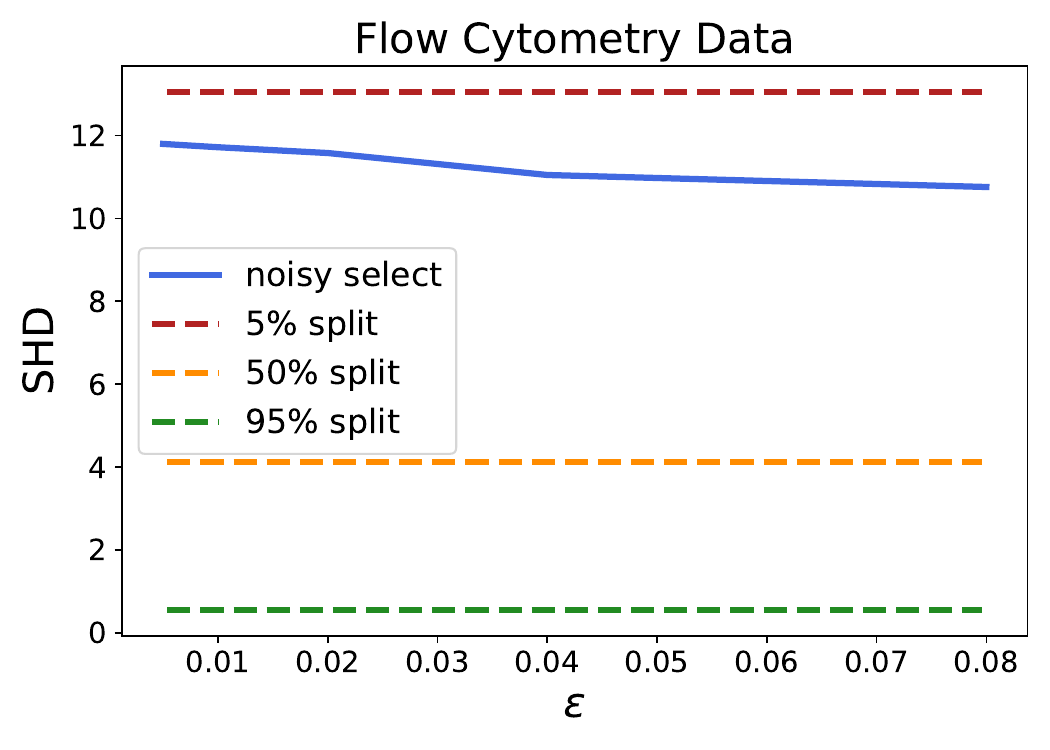}\hspace{0.5cm}
         \includegraphics[width=0.4\textwidth]{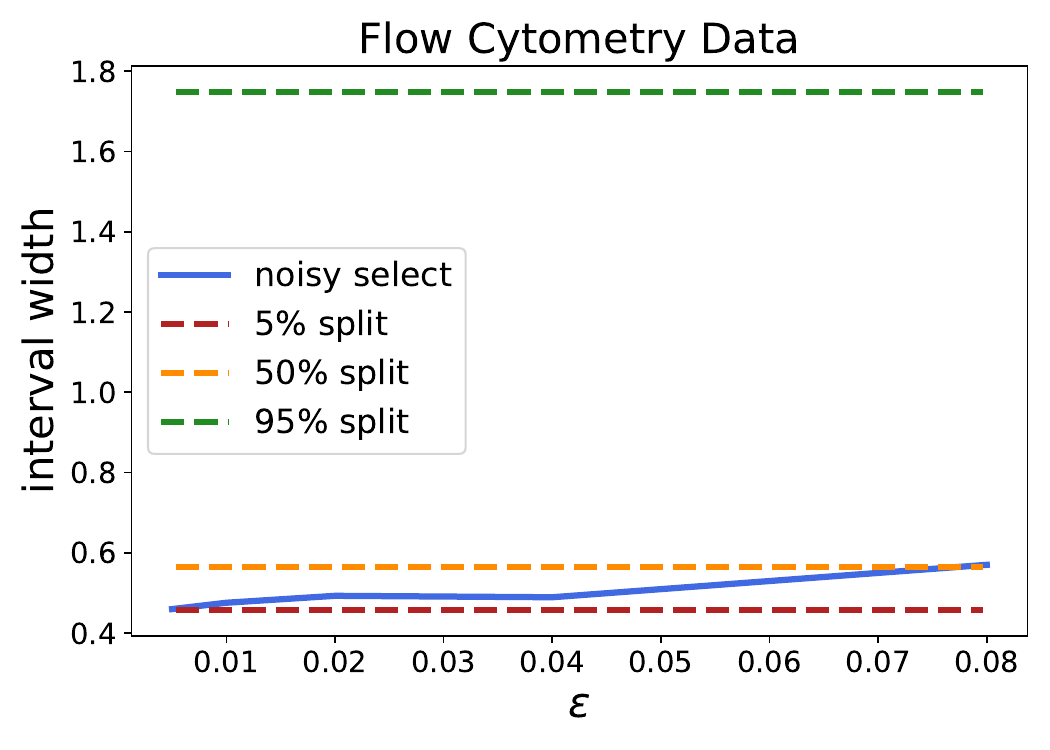}
        \caption{Comparison of \textsc{noisy-select} with varying $\epsilon$ and three data-splitting baselines in terms of SHD (left) and interval widths (right) on the flow cytometry data from \cite{sachs2005causal}.}
        \label{fig:sachs_huber}
\end{figure}

\subsubsection{Greedy Search}

Finally, we investigate the behavior of greedy search, comparing the quality of the graph found via data splitting to the quality of the graph output by \textsc{noisy-ges} and the corresponding confidence interval widths.
We consider a graph-generating model from the prior sections---the Erd\H{o}s-R\'enyi random graph. We take $d = 15$, set the edge weights to value between 2 and 4 and the average degree to 1 as before. We let 
 $E_{\max} = 5$ and vary the cumulative $\epsilon$ from Lemma~\ref{lemma:GES_dp_guarantee} between $0.01$ to $0.32$.

In Figure \ref{fig:split_comparison_ges} we compare the quality of the graph found via GES with three data splitting baselines and the graph found via \textsc{noisy-ges},  in terms of the structural Hamming distance (SHD) to the graph found by using GES with the same score and full data. We see that our method has better quality at the expense of wider intervals as we increase $\epsilon$ and that it can achieve behavior similar to a multitude of data split choices depending on the noise parameter selected.

\begin{figure}[b]
     \centering
         \includegraphics[width=0.4\textwidth]{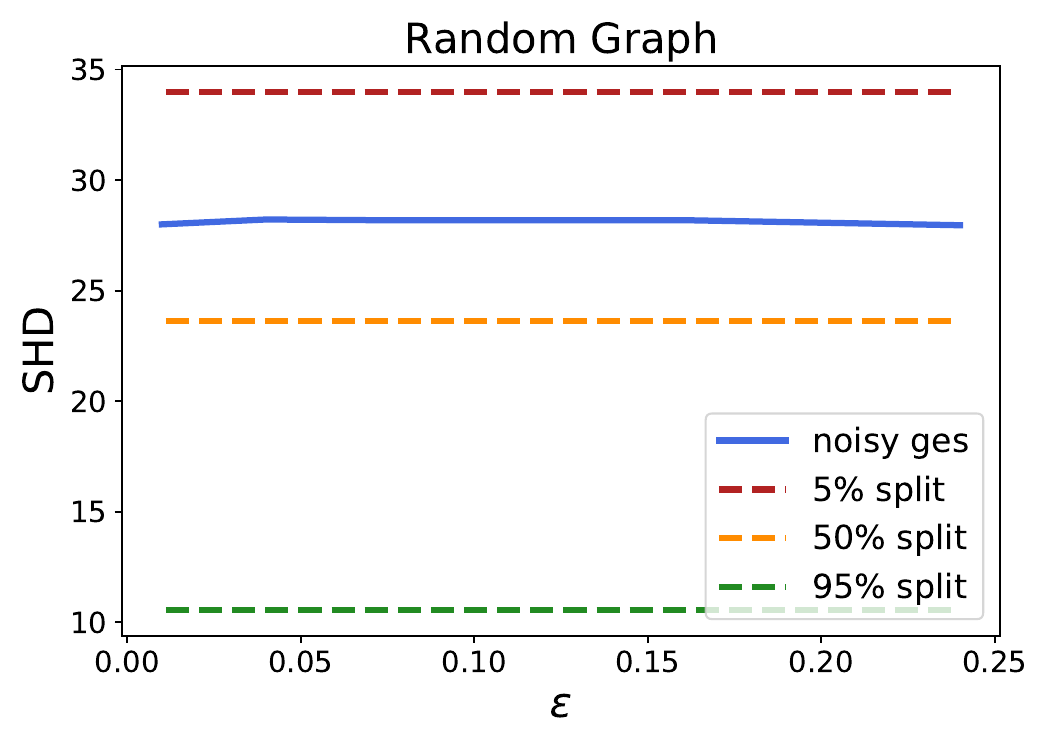}\hspace{0.5cm}
    \includegraphics[width=0.4\textwidth]{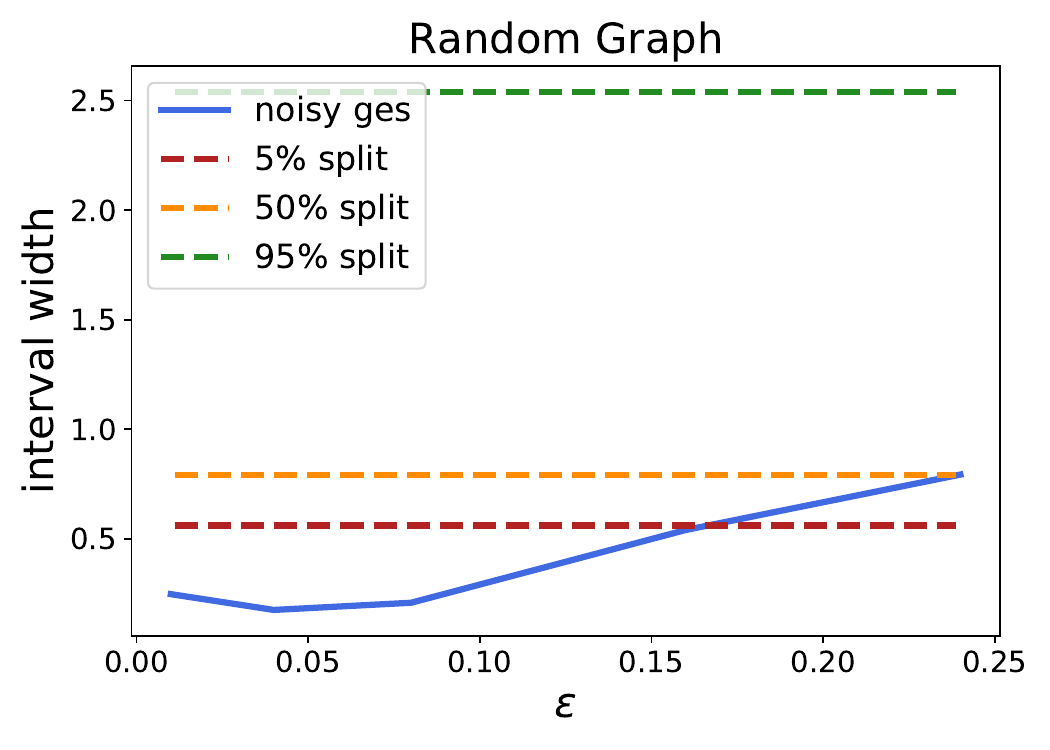}
     \caption{Comparison of \textsc{noisy-ges} with varying cumulative $\epsilon$ and classical GES with three data splitting choices in terms of SHD (left) and interval widths (right) for a random graph.}
        \label{fig:split_comparison_ges}
\end{figure}

\section{Discussion}
\label{sec:discussion}

We have presented tools for rigorous uncertainty quantification for causal estimands after score-based causal discovery. We discuss several extensions of our work and future directions.

\subsection{Causal Discovery with Interventional Data}
\label{sec:interventional}

We focused on providing valid inferences after causal graph discovery from observational data. However, our tools can be readily extended to settings 
where the data comes from multiple interventional distributions, as opposed to a single observational 
distribution.
Indeed, the principles underlying our correction do not fundamentally rely on the data being
i.i.d., and are applicable whenever the data points are merely independent---as in the
case of data collected from a set of independent interventions. We
sketch this more general setup and explain how
\textsc{noisy-ges} can be adapted to obtain a
randomized version of \emph{greedy interventional equivalence search}
(GIES)~\citep{hauser2012characterization}, a counterpart of
GES that operates on interventional data.

Suppose the data set $\D$ consists of
$n$ independent draws, $\{(X^{(i)}, \mathcal{T}^{(i)})\}_{i=1}^n$,
where $X^{(i)}$ denotes a data point from the intervention
described by $\mathcal{T}^{(i)}$. Here, $\mathcal{T}^{(i)}$ specifies
which variables are intervened on and from which distribution their
values are sampled; see \cite{hauser2012characterization} for a formal
description of how interventions can be encoded in
$\mathcal{T}^{(i)}$.

The result of \Cref{prop:max-info}, which translates differential
privacy to a bound on the max-information, only requires that the data
points are independent; the same max-information bound is thus true for
the interventional setting. Therefore, given a
differentially private version of GIES, we can straightforwardly
obtain a finite-sample correction to classical confidence intervals,
as in \Cref{thm:ges_inf}.


To finalize the argument, we note that the noise addition strategy in
\Cref{alg:GES_general} immediately ensures differential privacy of GIES. The reason is
that GES and GIES perform computations on the data in essentially the
same way; the only differences between GES and GIES are
graph-theoretic and do not alter the basic computations applied to the
data (in particular, GES maintains a CPDAG while GIES maintains a
so-called interventional essential graph, and the set of
valid insertion and deletion operators is different for the two
algorithms). As a result, the randomization mechanism
in \Cref{alg:GES_general} carries over directly to GIES without any modifications; it has the same privacy guarantee and
leads to the same procedure for valid downstream statistical
inference.

\subsection{Inference under Misspecification}

We provide rigorous inferential guarantees on ``projection'' parameters, such as the least-squares regression coefficients in Eq.~\eqref{eq:causal_effect}, in a manner akin to existing literature on post-selection inference. Naturally, the projection parameter may be far from the true causal effect if the discovered graph is incorrect (interestingly, though, it may coincide with the true effect even if the graph is incorrect, such as in the empty graph setting in Section \ref{sec:experimental}). Importantly, this misspecification of the inferential target is not an artifact of algorithmic causal discovery---even if the graph comes from domain knowledge, it may be misspecified and the same issue arises. In fact, one may argue that this problem is inevitable; real-world mechanisms are too complex to be represented accurately in detail, and experts are often biased toward choosing simple working models. For this reason, it is necessary to not only have concepts and tools for valid inference under correctly specified settings, but also for misspecified ones. We see this as one of this work's contributions: introducing notions of valid statistical inference under misspecified causal modeling.

Nevertheless, going forward it would be valuable to understand how graph misspecification propagates into misspecification of the inferential target. We begin this investigation with additional experiments in Section \ref{sec:true_effect}, where we evaluate the validity of our methods with respect to the \emph{true} causal effect, focusing on the same simulation settings as in Section~\ref{sec:experimental}. Relatedly, it would be valuable to gain insight into which graph structures and forms of graph misspecification are benign for specification of the inferential target and which are severe. We conjecture that sparse causal structures should be robust to misspecification of the target arising from the projection perspective (which is consistent with our experimental results in Section \ref{sec:true_effect}).

\subsection{Other Causal Discovery Methods}

Finally, we believe that our tools can be extended to other causal discovery methods, in particular constraint-based, continuous-optimization-based, and hybrid approaches. The main technical challenge is to ensure that the respective methods are differentially private; the subsequent application of max-information bounds and implications to statistical validity would remain the same as in our work. 
We believe in particular that existing algorithms for differentially private optimization \citep{abadi2016deep} could be applied toward making continuous-optimization-based methods for causal discovery private.

\section*{Acknowledgements}

We thank Dan Malinsky, Ben Recht, and Bin Yu for several inspiring discussions and helpful feedback.
This work was supported in part by the Vannevar Bush Faculty Fellowship program
under grant number N00014-21-1-2941. Y. W. was supported in part by the Office of Naval Research under
grant number N00014-23-1-2590, the National Science Foundation under
Grant No. 2231174, No. 2310831, No. 2428059, and a Michigan Institute
for Data Science Propelling Original Data Science (PODS) grant.

\bibliographystyle{agsm}

\bibliography{jasa/POSI}

\newpage

\appendix

\section{Greedy Equivalence Search: Background}
\label{app:ges_appendix}

In this section we provide the details behind the greedy pass subroutine (Algorithm \ref{alg:GES_single_pass_general}) that is used in GES. In particular, we review the definitions of valid $(\texttt{sgn})$-operators that appear in \citet{chickering2002optimal}, clarify what it means to apply a given operator to the current CPDAG, and explain how the score gains $\Delta S^{\texttt{sgn}} (e,\widehat G, \D)$ are computed. As before, we use $\widehat G$ to denote the CPDAG maintained by GES.

Before we define $(\texttt{sgn})$-operators, we briefly review some graph-theoretic preliminaries. We say two nodes $X_a, X_b$ are \textit{neighbors} in a CPDAG $\widehat G$ if they are connected by an undirected edge, and \textit{adjacent} if they are connected by any edge (directed or undirected). We also call a path from $X_a$ to $X_b$ in a CPDAG \textit{semi-directed} if each edge along it is either undirected or directed away from $X_a$.

\begin{definition}
For non-adjacent $X_a$ and $X_b$ in $\widehat G$, and a subset $\mathbf{T}$ of $X_b$'s neighbors that are not adjacent to $X_a$, the $\mathrm{Insert}(X_a, X_b, \mathbf{T})$ operator is defined as the procedure that modifies $\widehat G$ by:
\begin{enumerate}
    \item inserting edge $X_a\rightarrow X_b$;
    \item for each $T \in \mathbf{T}$, converting $T-X_b$ to $T\rightarrow X_b$.
\end{enumerate}
\end{definition}

\begin{definition}
For $X_a$ and $X_b$ in $\widehat G$ connected as $X_a-X_b$ or $X_a\rightarrow X_b$, and a subset $\mathbf{T}$ of $X_b$'s neighbors that are adjacent to $X$, the $\mathrm{Delete}(X_a, X_b, \mathbf{T})$ operator is defined as the procedure that modifies $\widehat G$ by:
\begin{enumerate}
    \item deleting the edge between $X_a$ and $X_b$, 
    \item for each $T \in \mathbf{T}$, converting $X_b - T$ to $X_b \rightarrow T$ and $X_a-T$ to $X_a\rightarrow T$.
\end{enumerate}
\end{definition}

We use ``$(+)$-operator'' (resp. ``$(-)$-operator'') as a shorthand for the $\mathrm{Insert}$ operator (resp. the $\mathrm{Delete}$ operator).

Now that we have a definition of $(\texttt{sgn})$-operators, we need to define which operators are \emph{valid} to apply to the current graph. For example, if we were greedily updating only a single DAG and not a CPDAG, we would only consider edge additions that maintain the DAG structure. We define an analogous form of validity for CPDAGs, which requires a bit more care. Let $\mathbf{NA}_{X_b,X_a}$ be the neighbors of $X_b$ that are adjacent to $X_a$.

\begin{definition}
\label{def:valid_+}
We say that $\mathrm{Insert}(X_a, X_b, \mathbf{T})$ is \textit{valid} if:
\begin{enumerate}
    \item $\mathbf{NA}_{X_b, X_a} \cup \mathbf{T}$ is a clique,
    \item Every semi-directed path from $X_b$ to $X_a$ contains a node in $\mathbf{NA}_{X_b, X_a} \cup T$.
\end{enumerate}
\end{definition}

\begin{definition}
\label{def:valid_-}
We say that $\mathrm{Delete}(X_a, X_b, \mathbf{T})$ is \textit{valid} if $\mathbf{NA}_{X_b, X_a} \setminus \mathbf{T}$ is a clique.
\end{definition}

For a valid $(\texttt{sgn})$-operator, Chickering also defines how to properly score the gain due to applying it. 
In particular, the gain due to executing $\mathrm{Insert}(X_a, X_b, \mathbf{T})$ is defined as: 
\begin{equation}
\label{eq:score_gain_CPDAG1}
\Delta S^+((X_a, X_b, \mathbf{T}), \widehat G, \D) = s(X_a, \mathbf{NA}_{X_b, X_a} \cup \mathbf{T} \cup \mathbf{Pa}_{X_b} \cup X_a, \D) - s(X_b, \mathbf{NA}_{X_b, X_a} \cup \mathbf{T} \cup \mathbf{Pa}_{X_b} \cup X_a, \D).
\end{equation}
This expression is essentially an application of the identity shown in Section \ref{sec:ges}, $\Delta S^+(e, G, \D) = s(X_j,\mathbf{Pa}_j^{G}\cup X_i, \D) - s(X_j,\mathbf{Pa}_j^{G}, \D)$, for a specific DAG $G$ consistent with the CPDAG $\widehat G$ and edge $e = (X_a \rightarrow X_b)$.
Similarly, the score gain due to executing $\mathrm{Delete}(X_a, X_b, \mathbf{T})$ is defined as: 
\begin{equation}
\label{eq:score_gain_CPDAG2}	
\Delta S^{-} ((X_a, X_b, \mathbf{T}), \widehat G, \D) = s(X_b, \{\mathbf{NA}_{X_b, X_a} \setminus \mathbf{T}\} \cup \{\mathbf{Pa}_{X_b} \setminus X_a\}, \D) - s(X_b, \{\mathbf{NA}_{X_b, X_a} \setminus \mathbf{T}\} \cup \mathbf{Pa}_{X_b}, \D).
\end{equation}

Having laid out this preamble, we can now state more precisely the greedy pass subroutine (Algorithm~\ref{alg:GES_single_pass_general}) of noisy GES, which we do in Algorithm~\ref{alg:GES_single_pass_full}.

\begin{algorithm}[H]
\SetAlgoLined
\SetKwInOut{Input}{input}
\Input{initial graph $\widehat G_0$, data set $\D$, maximum number of edges $E_{\max}$, score $S$ with local score sensitivity~$\tau$, privacy parameters $\epsilon_{\mathrm{score}}, \epsilon_{\mathrm{thresh}}$, pass indicator $\texttt{sgn}\in\{+,-\}$}
\textbf{output:} estimated causal graph $\widehat G$\newline
Initialize $\widehat G\leftarrow \widehat G_0$\newline
Sample noisy threshold $\nu \sim \text{Lap}\left(\frac{4 \tau}{\epsilon_{\mathrm{thresh}}}\right)$\newline
\For{$t=1,2,\dots,E_{\max}$}{
\uIf{$\texttt{sgn} = +$}
{
Construct set $\mathcal{E}^+_t$ of all valid $\mathrm{Insert}(X_a,X_b,\mathbf{T})$ operators (Def.~\ref{def:valid_+})
}
\ElseIf{$\texttt{sgn} = -$}{
Construct set $\mathcal{E}^-_t$ of all valid $\mathrm{Delete}(X_a,X_b,\mathbf{T})$ operators (Def.~\ref{def:valid_-})
}
\ For all $e\in \mathcal{E}^{\texttt{sgn}}_t$, compute $\Delta S^{\texttt{sgn}}(e, \widehat{G}, \D)$ (according to \Cref{eq:score_gain_CPDAG1} or \eqref{eq:score_gain_CPDAG2}) and sample $\xi_{t,e} \stackrel{\text{i.i.d.}}{\sim} \text{Lap}\left(\frac{4 \tau}{\epsilon_{\mathrm{score}}}\right)$\newline
Set $e^*_t = \argmax_{e\in \mathcal{E}^{\texttt{sgn}}_t} \Delta S^{\texttt{sgn}}(e, \widehat{G}, \D) + \xi_{t,e}$\newline
Sample $\eta_t \sim \text{Lap}\left(\frac{8 \tau}{\epsilon_{\mathrm{thresh}}}\right)$\newline
\uIf{$\Delta S^{\texttt{sgn}}(e^*_t, \widehat G, \D) + \eta_t > \nu$}
{
Apply operator $e_t^*$ to $\widehat G$
}
\Else{
break
}
}
Return $\widehat G$
\caption{GreedyPass}
\label{alg:GES_single_pass_full}
\end{algorithm}

\section{Noisy Causal Discovery: Proofs}\label{sec:noisy_discovery_proofs}

\subsection{Proof of Lemma \ref{lemma:dp_noisy_sel}}

The proposition is an application of the privacy guarantees of the Report Noisy Max mechanism in differential privacy~\citep[see, e.g.,][Chapter 3.3]{dwork2014algorithmic}. In addition, the privacy analysis of Algorithm \ref{alg:GES_general} strictly subsumes the privacy analysis of Algorithm \ref{alg:score-based-sel}.

\subsection{Proof of Theorem \ref{thm:score_fn}}

By Proposition \ref{prop:max-info}, we can bound the max-information between $\widehat G$ and $\D$:
$$\mathcal{I}_\infty^{\gamma}(\widehat G;\D) \leq \frac{n}{2}\epsilon^2 + \epsilon\sqrt{n\log(2/\gamma)/2}.$$
The definition of max-information, in turn, implies that 
\begin{align*}
&\PP{\exists (i,j)\in\mathcal{I}_{G}:\beta_G^{(i\rightarrow j)} \not\in \mathrm{CI}_G^{(i\rightarrow j)}(\tilde\alpha), \widehat G = G}\\
&\quad \leq \exp\left(\mathcal{I}_\infty^\gamma(\widehat G;\D)\right) \PP{\exists (i,j)\in\mathcal{I}_{G}:\beta_G^{(i\rightarrow j)} \not\in \mathrm{CI}_G^{(i\rightarrow j)}(\tilde\alpha; \tilde\D), \widehat G = G} + \gamma\\
&\quad \leq \exp\left(\frac{n}{2}\epsilon^2 + \epsilon\sqrt{n\log(2/\beta)/2}\right)\tilde \alpha +\gamma\\
&\quad = \alpha.
\end{align*}
Marginalizing over all graphs $G$ yields the final theorem statement.

\subsection{Proof of Proposition \ref{prop:noisy_discovery_utility}}

Fix $\delta\in(0,1)$. Then, for any graph $G\in\mathcal{G}$ with $S(G,\D) \leq S(\widehat G_*,\D) - \frac{4\tau}{\epsilon}\log(2/\delta)$, noisy graph discovery outputs $G$ with probability at most $\delta$.
 
Suppose that \textsc{noisy-select} outputs a graph $G$ which is at least $\frac{4\tau}{\epsilon}\log(2/\delta)$ suboptimal, i.e. $S( G,\D) \leq S(\widehat G_*,\D) - \frac{4\tau}{\epsilon}\log(2/\delta)$. Then, this means that at least one of the following must be true: $\xi_{G} \geq \frac{2\tau}{\epsilon}\log(2/\delta)$ or $\xi_{\widehat G_*} \leq -\frac{2\tau}{\epsilon}\log(2/\delta)$. Using the CDF of the Laplace distribution together with a union bound, we have that 
$$\PP{\widehat G = G} \leq \PP{\xi_{G} \geq \frac{2\tau}{\epsilon}\log(2/\delta) \cup \xi_{\widehat G_*} \leq -\frac{2\tau}{\epsilon}\log(2/\delta)} \leq 2\PP{\xi_{G} \geq \frac{2\tau}{\epsilon}\log(2/\delta)} = \delta.$$



\section{Noisy Greedy Equivalence Search: Proofs}

\subsection{Differential Privacy Preliminaries}

\begin{lemma}[Closure under post-processing~\citep{dwork2006calibrating}]
\label{lemma:dp_postprocessing}
Let $\A(\cdot)$ be an $\epsilon$-differentially private algorithm and let $\mathcal B$ be an arbitrary, possibly randomized map. Then, $\mathcal B\circ \A(\cdot)$ is $\epsilon$-differentially private.	
\end{lemma}

\begin{lemma}[Adaptive composition~\citep{dwork2006calibrating}]
\label{lemma:dp_composition}
For $t\in[k]$, let $\A_t(\cdot,a_1,a_2,\dots,a_{t-1})$ be $\epsilon_t$-differentially private for all fixed $a_1,\dots,a_{t-1}$. Then, the algorithm $\A_{\mathrm{comp}}$ which executes $\A_1,\dots,\A_k$ in sequence and outputs $a_1 = \A_1(\D), a_2 = \A_2(\D,a_1),\dots,a_k = \A_k(\D,a_1,\dots,a_{k-1})$ is $(\sum_{t=1}^k \epsilon_t)$-differentially private.
\end{lemma}

\subsection{Proof of Lemma \ref{lemma:GES_dp_guarantee}}

As mentioned earlier, the proof relies on the analysis of two differentially private mechanisms: Report Noisy Max and Above Threshold \citep{dwork2014algorithmic}. To facilitate the proof, in Algorithm \ref{alg:GES_single_pass_decoupled} we provide an equivalent reformulation of Algorithm \ref{alg:GES_single_pass_general} that allows decoupling the analyses of these two mechanisms.

\begin{algorithm}[H]
\SetAlgoLined
\SetKwInOut{Input}{input}
\Input{initial graph $\widehat G_0$, data set $\D$, maximum number of edges $E_{\max}$, score $S$ with local score sensitivity $\tau$, privacy parameters $\epsilon_{\mathrm{score}}, \epsilon_{\mathrm{thresh}}$, pass indicator $\texttt{sgn}\in\{+,-\}$}
\textbf{output:} estimated causal graph $\widehat G$ \newline
Initialize $\widehat G \leftarrow \widehat G_0$\newline
Get potential operators $\mathcal{E} \leftarrow \text{ProposeOperators}(\widehat G, \D, E_{\max}, S, \tau, \epsilon_{\mathrm{score}}, \texttt{sgn})$ \newline
Get selected operator subset $\mathcal{E}^* \leftarrow \text{SelectOperators}(\widehat G, \D, S, \tau, \epsilon_{\mathrm{thresh}}, \texttt{sgn}, \mathcal E)$ \newline
\For{$t = 1, \ldots, |\mathcal E^*|$}{Apply $e_t^*$ to $\widehat G$}
Return $\widehat G$
\caption{Decoupled GreedyPass}
\label{alg:GES_single_pass_decoupled}
\end{algorithm}

\begin{algorithm}[H]
\SetAlgoLined
\SetKwInOut{Input}{input}
\Input{initial graph $\widehat G_0$, data set $\D$, maximum number of edges $E_{\max}$, score $S$ with local score sensitivity $\tau$, privacy parameter $\epsilon_{\mathrm{score}}$, pass indicator $\texttt{sgn} \in \{+, -\}$}
\textbf{output:} proposed set of operators $\mathcal E$\newline
Initialize $\widehat G \leftarrow \widehat G_0$\newline
Initialize $\mathcal E \leftarrow \emptyset$\newline
\For{$t=1,2,\dots,E_{\max}$}{
\ Construct set $\mathcal{E}^{\mathrm{sgn}}_t$ of valid $(\texttt{sgn})$-operators \newline
For all $e\in \mathcal{E}_t^{\mathrm{sgn}}$, compute $\Delta S^{\mathrm{sgn}}(e, \widehat G, \D)$ and sample $\xi_{t,e} \stackrel{\text{i.i.d.}}{\sim} \text{Lap}\left(\frac{4 \tau}{\epsilon_{\mathrm{score}}}\right)$\newline
Set $e_t = \argmax_{e\in \mathcal{E}_t^{\mathrm{sgn}}} \Delta S^{\mathrm{sgn}}(e, \widehat G, \D) + \xi_{t,e}$\newline
Add operator $e_t$ to $\mathcal E$ \newline
Apply operator $e_t$ to $\widehat G$
}
Return $\mathcal E = (e_1,\dots,e_{E_{\max}})$
\caption{ProposeOperators}
\label{alg:propose_operators}
\end{algorithm}

\begin{algorithm}[H]
\SetAlgoLined
\SetKwInOut{Input}{input}
\Input{initial graph $\widehat G_0$, data set $\D$, score $S$ with local score sensitivity $\tau$, privacy parameter $\epsilon_{\mathrm{thresh}}$, pass indicator $\texttt{sgn} \in \{+, -\}$, set of proposed operators $\mathcal E$}
\textbf{output:} set of operators $\mathcal E^*$\newline
Sample noisy threshold $\nu \sim \text{Lap}\left(\frac{4 \tau}{\epsilon_{\mathrm{thresh}}}\right)$\newline
Initialize $\mathcal E^* \leftarrow \emptyset$\newline
Initialize $\widehat G \leftarrow \widehat G_0$\newline
\For{$t=1,2,\dots,|\mathcal E|$}{
Sample $\eta_t \sim \text{Lap}\left(\frac{8 \tau}{\epsilon_{\mathrm{thresh}}}\right)$\newline
\uIf{$\Delta S^{\mathrm{sgn}}(e_t, \widehat G, \D) + \eta_t \geq \nu$}{
\ Add $e_t^*$ to $\mathcal E^*$\newline
Apply $e_t^*$ to $\widehat G$}
\Else{
break
}
}
Return $\mathcal{E}^* = (e_1^*,e_2^*,\dots)$
\caption{SelectOperators}
\label{alg:select_operators}
\end{algorithm}

We argue that the two subroutines composed in the greedy pass, namely ProposeOperators (Algorithm~\ref{alg:propose_operators}) and SelectOperators (Algorithm~\ref{alg:select_operators}), are differentially private. By the closure of differential privacy under post-processing (Lemma \ref{lemma:dp_postprocessing}), this will imply that Algorithm \ref{alg:GES_single_pass_decoupled}, which returns $\widehat G$, is also differentially private, since $\widehat G$ is merely a post-processing of the selected operators $\mathcal{E}^*$.

The privacy guarantee of Algorithm \ref{alg:propose_operators} is implied by the usual privacy guarantee of Report Noisy Max and composition of differential privacy. Note that the construction of the set $\mathcal{E}_t^{\mathrm{sgn}}$ at every time step is only a function of the current graph $\widehat G$ and not of the data; i.e., it is independent of the data conditioned on $\widehat G$. Formally, the key component is the following lemma:

\begin{lemma}\label{lemma:report_noisy_max}
For any $t\in[E_{\max}]$, selecting $e_t$ is $\epsilon_{\mathrm{score}}$-differentially private; that is, for any operator $e_0\in\mathcal{E}_t^{\texttt{sgn}}$, it holds that
$$\PPst{\argmax_{e\in \mathcal{E}_t^{\texttt{sgn}}} \Delta S^{\texttt{sgn}}(e, \widehat G, \D) + \xi_{t,e} = e_0}{\widehat G} \leq e^{\epsilon_{\mathrm{score}}}\PPst{\argmax_{e\in \mathcal{E}_t^{\texttt{sgn}}} \Delta S^{\texttt{sgn}}(e, \widehat G, \D') + \xi_{t,e} = e_0}{\widehat G},$$
for any current graph $\widehat{G}$ and any two data sets $\D, \D'$ that differ in at most one data point.
\end{lemma}

\begin{proof}
Denote $r_e \doteq \Delta S^{\texttt{sgn}} (e, \widehat G, \D)$ and $r'_e \doteq \Delta S^{\mathrm{sgn}} (e, \widehat G, \D')$. For a fixed $e_0\in\mathcal E_t^{\texttt{sgn}}$, define
$$\xi^\star_{t,e_0} \doteq \min\{\xi : r_{e_0} + \xi > r_{e'} + \xi_{t,e'} \; \forall e' \neq e_0\}.$$
For fixed $\{\xi_{t,e'}\}_{e'\neq e_0}$, we have that $e_0$ will be the selected operator on $\D$ if and only if $\xi_{t,e_0} \geq \xi^\star_{t,e_0}$.

Further, by the bounded sensitivity of the local scores, we have that for all $e' \neq e_0$:
\begin{align*}
r_{e_0} + \xi^\star_{t,e_0} > r_{e'} + \xi_{t,e'}  \\
\Rightarrow {r}'_{e_0} + 2\tau + \xi^\star_{t,e_0} > {r}'_{e'} - 2\tau + \xi_{t,e'} \\
\Rightarrow r'_{e_0} + \left(4\tau + \xi^\star_{t,e_0}\right) > r_{e'} + \xi_{t,e'}.
\end{align*}
Therefore, as long as $\xi_{t,e_0} \geq 4\tau + \xi^\star_{t,e_0}$, the selection on $\D'$ will be $e_0$ as well. Using the form of the density of $\xi_{t,e_0} \sim \text{Lap}\left(\frac{4\tau}{\eps_{\mathrm{score}}}\right)$, we have that:
\begin{align*}
\PPst{\argmax_{e\in \mathcal{E}_t^{\texttt{sgn}}} r_{e}' + \xi_{t,e} = e_0}{\{\xi_{t,e'}\}_{e'\neq e_0},\,\widehat G} &\geq \PP{\xi_{t,e_0} \geq 4\tau + \xi_{t,e_0}^\star} \\
&= e^{-\epsilon_{\mathrm{score}}} \PP{\xi_{t,e_0} \geq \xi_{t,e_0}^\star} \\
&= \PPst{\argmax_{e\in \mathcal{E}_t^{\texttt{sgn}}} r_{e} + \xi_{t,e} = e_0}{\{\xi_{t,e'}\}_{e'\neq e_0},\,\widehat G}.
\end{align*}
By taking iterated expectations, we have
$$\PPst{\argmax_{e\in\mathcal E_t^{\texttt{sgn}}} \Delta S^{\texttt{sgn}}(e, \widehat G, \D) + \xi_{t,e} = e_0}{\widehat G} \leq e^{\epsilon_{\mathrm{score}}}\PPst{\argmax_{e\in\mathcal E_t^{\texttt{sgn}}} \Delta S^{\texttt{sgn}}(e, \widehat G, \D') + \xi_{t,e} = e_0}{\widehat G},$$
 for all data sets $\D,\D'$ differing in at most one data point, ensuring the desired privacy.
\end{proof}

This directly implies the following result:

\begin{lemma}[Privacy of ProposeOperators]\label{lemma:propose_operators_privacy} Algorithm \ref{alg:propose_operators} is $E_{\max}\eps_{\mathrm{score}}$-differentially private.
\end{lemma}

\begin{proof}
The result follows directly from Lemma \ref{lemma:report_noisy_max}, by applying the adaptive composition rule for differential privacy (Lemma \ref{lemma:dp_composition}) over $E_{\max}$ steps.
\end{proof}

Now we isolate the second component of the greedy pass---checking if the operator's contribution is positive. To analyze this component independently of the selection of potential operators, we consider Algorithm \ref{alg:select_operators} which receives a set of proposed operators $\mathcal E$ and outputs only the first $E_{\max}^*\leq E_{\max}$ of them which pass the noisy threshold test. Note that $E_{\max}^*$ is random and data-dependent.

In what follows we use $\mathcal E^*(\D)$ and $\mathcal E^*(\D')$ to denote the output of Algorithm \ref{alg:select_operators} on two data sets $\D, \D'$ that differ in at most one data point.

\begin{lemma}[Privacy of SelectOperators]
\label{lemma:abovethresh}
For any input set of proposed edges $\mathcal E = (e_1,\dots,e_{E_{\max}})$, Algorithm \ref{alg:select_operators} is $\epsilon_{\mathrm{thresh}}$-differentially private; that is, for any $1\leq k \leq E_{\max}+1$:
$$\PP{\mathcal E^*(\D) = (e_j)_{j< k}} \leq e^{\eps_{\mathrm{thresh}}} \PP{\mathcal E^*(\D') = (e_j)_{j< k}}$$
given any two data sets $\D, \D'$ that differ in at most one data point.
\end{lemma}

\begin{proof}
Fix $1\leq k\leq E_{\max}+1$ and consider $(e_1,\dots,e_k)$. Let $G_1, \ldots, G_k$ be the graphs resulting from the application of operators $e_t$ in sequence, starting from the initial graph $\widehat G_0$. Define $r_t = \Delta S^{\mathrm{sgn}}(e_t, G_{t-1}, \D)$ and $r'_t = \Delta S^{\mathrm{sgn}}(e_t, G_{t-1}, \D')$. Condition on $\eta_1, \ldots, \eta_{k-1}$ and define the following quantity that captures the minimal value of the noisy score gain up to time $k-1$:
$$g(\D) = \min_{i<k} \{r_i + \eta_i\},$$
and analogously for $\D'$:
$$g(\D') = \min_{i<k} \{r'_i + \eta_i\}.$$
Using these quantities we can directly express the probability of outputting exactly the first $k-1$ proposed operators. Breaking at the $k$-th step of the algorithm, we have:
\begin{align*}
\PP{\mathcal E^*(\D)=(e_j)_{j< k}}&= \PP{\nu \in (r_k + \eta_k, g(\D)]} \\
&= \int_{-\infty}^\infty \int_{-\infty}^\infty p_{\eta_k}(q) p_\nu(w) \mathbf{1}\{w\in (r_k + q, g(\D)]\} dq dw.
\end{align*}
With the change of variables $q' = q - g( \D) + g(\D') + r_k - r'_k$, $ w' = w + g(\D') - g({\D})$, we obtain
$$\mathbf{1}\{w\in (r_k + q, g(\D)]\} = \mathbf{1}\{w' + g(\D) - g(\D') \in ( q' + g(\D) - g(\D') + r'_k, g(\D)]\} = \mathbf{1}\{ w' \in ( r_k' + q', g(\D')]\}$$
and thus
\begin{align*}
&\PP{\mathcal E^*(\D)=(e_j)_{j< k}}\\
&= \int_{-\infty}^\infty \int_{-\infty}^\infty p_{\eta_k}( q' + g(\D) - g( \D') - r_k + r_k') p_\nu(w' - g( D') + g(\D)) \mathbf{1}\{ w'\in ( r_k' + q', g( \D')]\} d q' d w'.
\end{align*}
Observe that $r_t$ is $2\tau$-sensitive since the local scores are $\tau$-sensitive, and hence $g(\D)$ is $2\tau$-sensitive as well. This implies that $|q' - q| \leq 4\tau$, $|w' - w| \leq 2\tau$, so by the form of the Laplace density we have
$$p_{\eta_k}( q' + g(\D) - g( \D') - r_k +  r_k') \leq e^{\eps_{\mathrm{thresh}} /2}p_{\eta_k}(q'),~~ p_{\nu}( w' - g( D') + g(\D)) \leq e^{\eps_{\mathrm{thresh}} /2}p_{\nu}( w').$$
Putting everything together, we have:
\begin{align*}
\PP{\mathcal E^*(\D)=(e_j)_{j<k}}
&\leq  \int_{-\infty}^\infty \int_{-\infty}^\infty e^{\eps_{\mathrm{thresh}} /2}p_{\eta_k}( q') p_\nu( w') e^{\eps_{\mathrm{thresh}}/2} \mathbf{1}\{ w' \in ( r_k' +  q', g(\D')]\} d q' d w' \\
&= e^{\eps_{\mathrm{thresh}}} \PP{\mathcal E^*(\D')=(e_j)_{j<k}},
\end{align*}
which is the desired guarantee.
\end{proof}

Finally, we combine the guarantees of Lemma \ref{lemma:report_noisy_max} and Lemma \ref{lemma:abovethresh} to infer the privacy parameter of Decoupled GreedyPass (Algorithm \ref{alg:GES_single_pass_decoupled}), which is equivalent to GreedyPass from Algorithm \ref{alg:GES_general}. The following statement follows from a direct application of privacy composition (i.e., Lemma \ref{lemma:dp_composition}).


\begin{lemma}[Privacy of Decoupled GreedyPass]\label{lemma:decoupled_ges_priv}
Algorithm \ref{alg:GES_single_pass_decoupled} is $\eps_{\mathrm{thresh}} + E_{\max}\eps_{\mathrm{score}}$-differentially private.
\end{lemma}

%

\begin{proof}[Proof of Lemma \ref{lemma:GES_dp_guarantee}] Since the GES algorithm (Algorithm \ref{alg:GES_general}) consists of two executions of GreedyPass, which is equivalent to the Decoupled GreedyPass, we can apply Lemma \ref{lemma:decoupled_ges_priv} and Lemma \ref{lemma:dp_composition} to conclude that GES is $(2\epsilon_{\mathrm{thresh}} + 2E_{\max} \epsilon_{\mathrm{score}})$-differentially private.
\end{proof}

\section{Further Empirical Studies and Details}\label{app:experiments}

\subsection{Additional Experimental Details}

\paragraph{Estimation of Sensitivity.} For \textsc{noisy-select} we estimate the sensitivity at every step by computing the maximum score change between the original dataset and a copy which has one entry replaced by another randomly selected one from the same dataset. We take the maximum over the graph variants that the method chooses between. For \textsc{noisy-ges} we perform a similar procedure but at every step we take the maximum score change over ten randomly chosen possible edge additions/deletions.

\paragraph{Huber score.} Throughout our experiments we use a variation of the BIC score where we replace the squared loss with the Huber loss for increased robustness. Informally, we refer to this score as the ``Huber score.'' More precisely, for a fixed $\delta>0$, we use
\begin{equation*}
S_{\text{Huber}}\left(G,\D\right) = - \min_\theta \frac{1}{n \sigma^2} \sum_{j=1}^d \sum_{k=1}^n L_\delta\left(X_j^{(k)} - \sum_{s \in \mathbf{Pa}_j^{G}}\theta_s X_s^{(k)} \right) - \sum_{j=1}^d\frac{|\mathbf{Pa}_j^G|}{n} \log n.
\end{equation*}
where $L_\delta(a) = \begin{cases} \frac{1}{2}a^2 &\text{if } a \leq \delta, \\\delta \cdot (|a|-\frac{1}{2}\delta) &\text{otherwise.}\end{cases}$\\

We set the $\delta$ parameter to $0.5$ in all experiments except for the Auto-MPG one where we use $\delta =0.25$. 
 
\subsection{Validity Experiments for Classical GES with Huber Score}\label{subsec:huber_validity}

In Figure~\ref{fig:validity_plots_huber} below we show the miscoverage probability in the settings from Section~\ref{sec:exp_validity} but this time using the Huber score for GES as well. This shows that simply using a robust score does not help ameliorate the miscoverage caused by ``double dipping.'' 
\begin{figure}[t]
     \centering
     \includegraphics[width=0.4\textwidth]{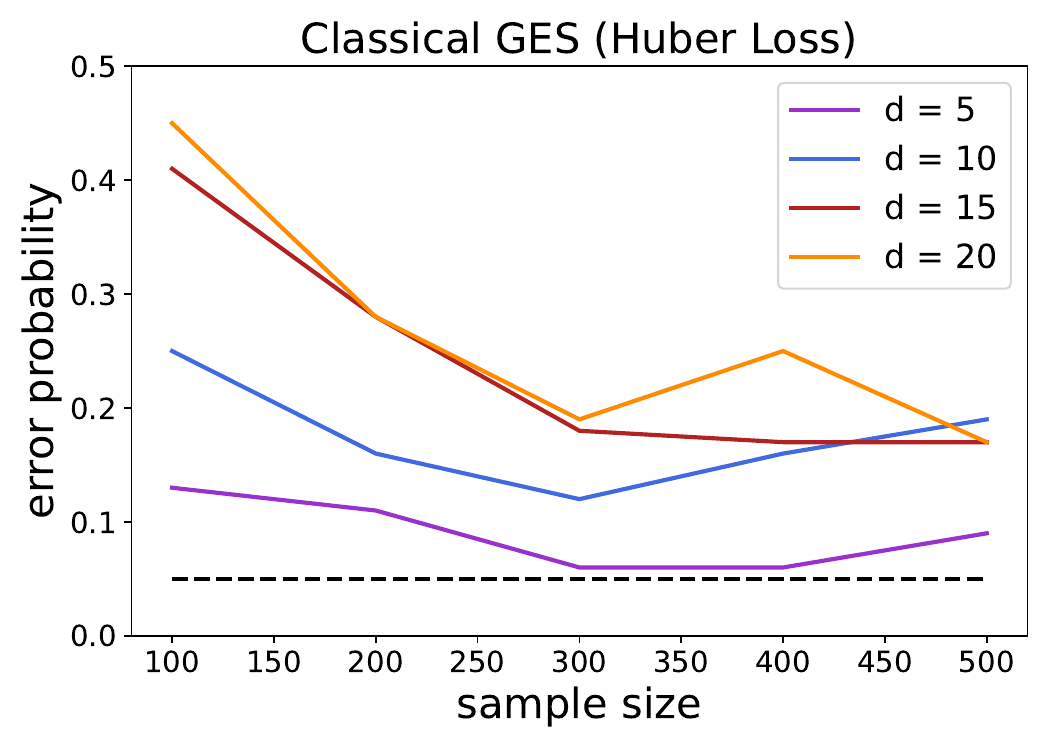}
    \hspace{0.5cm}\includegraphics[width=0.4\textwidth]{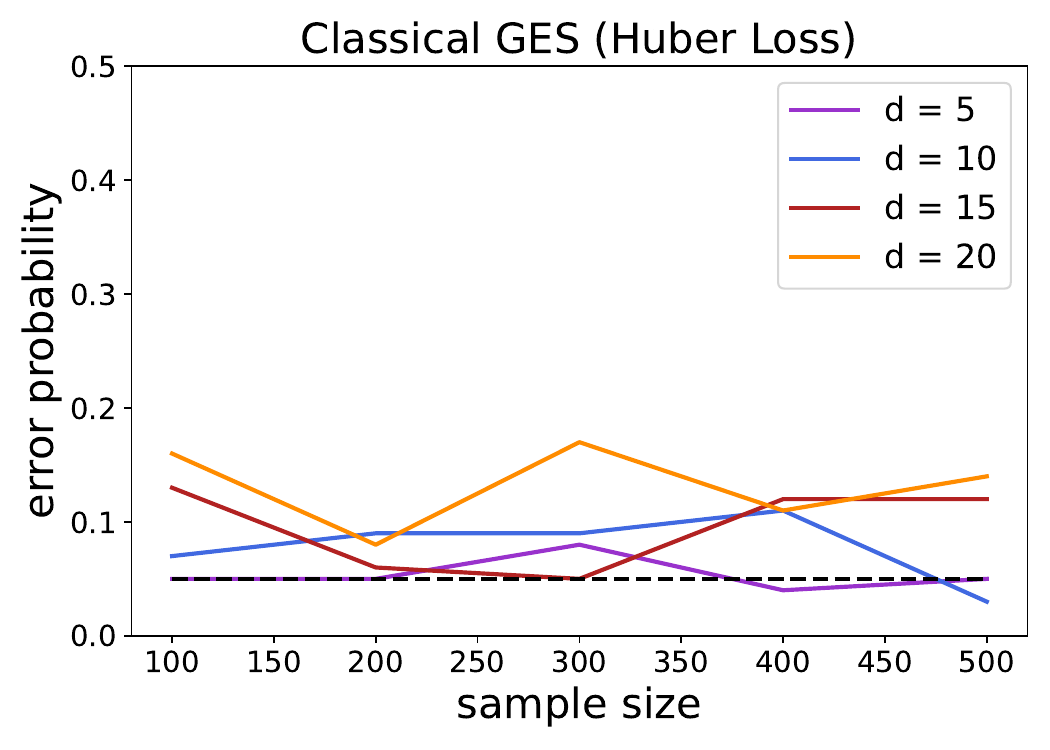}
        \caption{Probability of error for varying $n$ and $d$ of classical GES with Huber score in the empty graph (left) and random graph (right).}
        \label{fig:validity_plots_huber}
\end{figure}

\subsection{Validity with Respect to the True Causal Effect}
\label{sec:true_effect}

We perform additional experiments where we evaluate the validity of \textsc{noisy-ges} with respect to the true causal effect. We focus on the sparse random graph setting from Section~\ref{sec:exp_validity}, this time taking the true causal effect as the ground-truth value of the estimand, rather than the projection parameter. Recall that in the empty graph setting the plots in Section~\ref{sec:exp_validity} implicitly measured error with respect to the true causal effect.  We set $\alpha=0.05$ and vary $\epsilon\in\{0.02, 0.04\}$ and $d\in\{5,10,15,20\}$. We average the error over $100$ trials.

In Figure \ref{fig:true_effect_select} we plot the error of exact selection and \textsc{noisy-select}. We observe that exact selection violates the type I error guarantee for all $d$ and all sample sizes; the error is above $0.05$ almost everywhere. The \textsc{noisy-select} method exhibits an inflated error rate for smaller sample sizes, but drops the error below $\alpha$ for larger sample sizes. In Figure \ref{fig:true_effect_ges} we plot the error of classical GES and \textsc{noisy-ges}. Classical GES exhibits a significantly higher inflation of error than \textsc{noisy-ges}, showing the benefits of a post-selection correction. Moreover, the error of \textsc{noisy-ges} is not inflated uncontrollably: it essentially does not exceed $0.15$, supporting our conjecture that in sparse graphs projection effects are likely to be close to true causal effects.

\begin{figure}[t]
     \centering
     \includegraphics[width=0.32\textwidth]{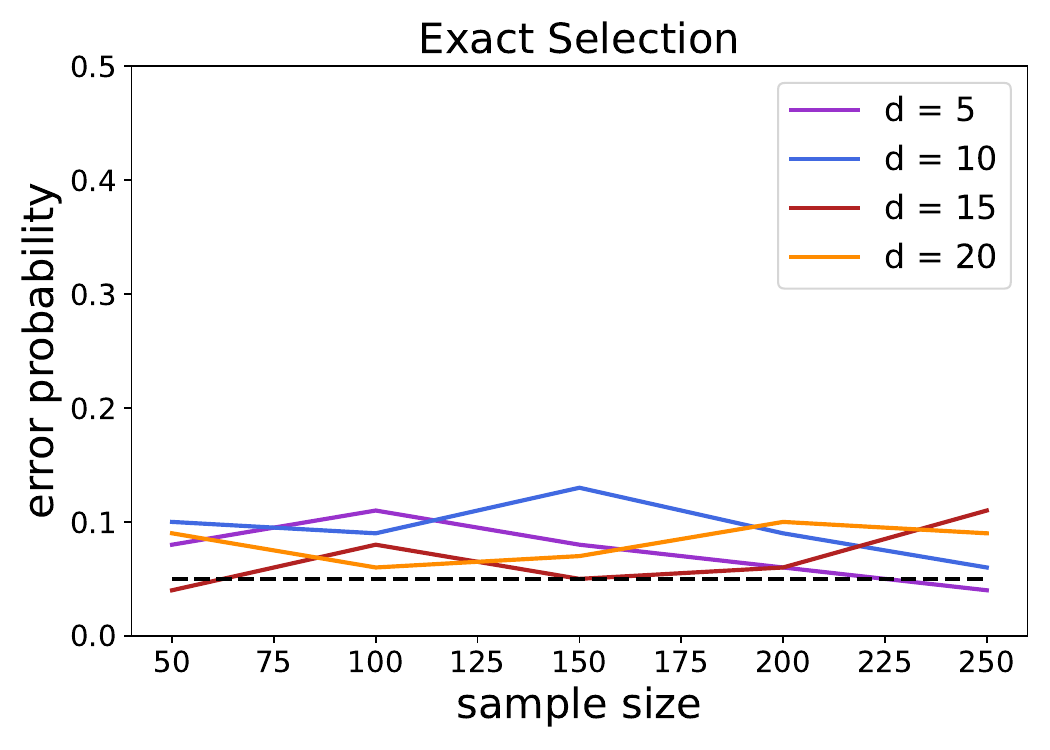}
     \includegraphics[width=0.32\textwidth]{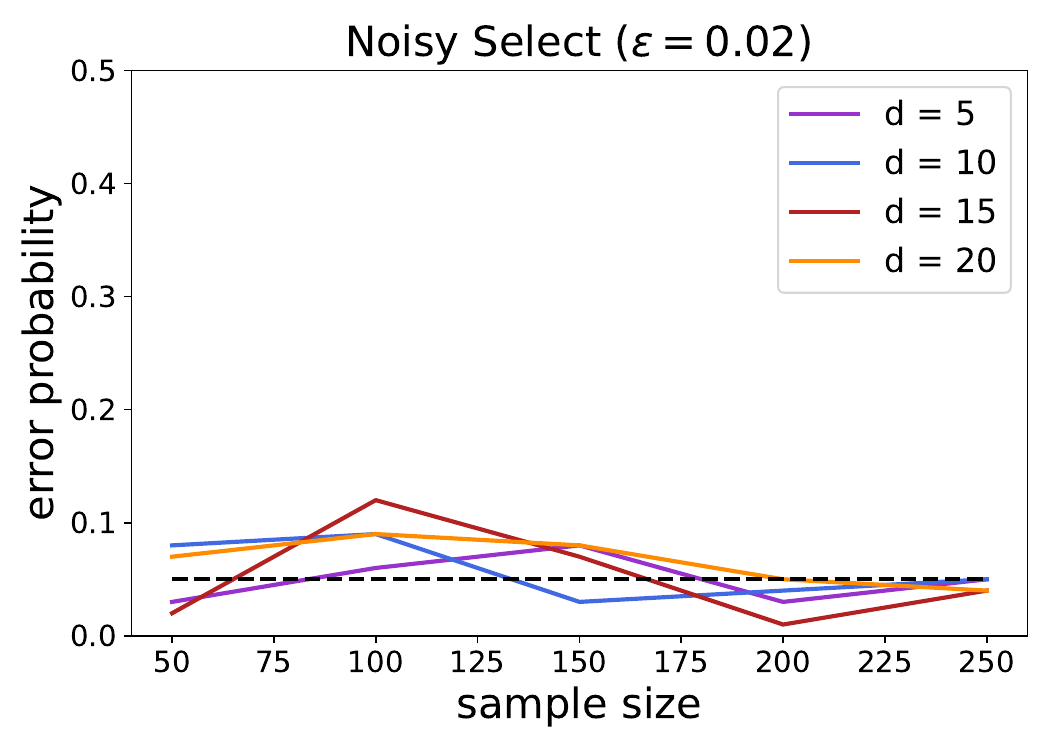}
    \includegraphics[width=0.32\textwidth]{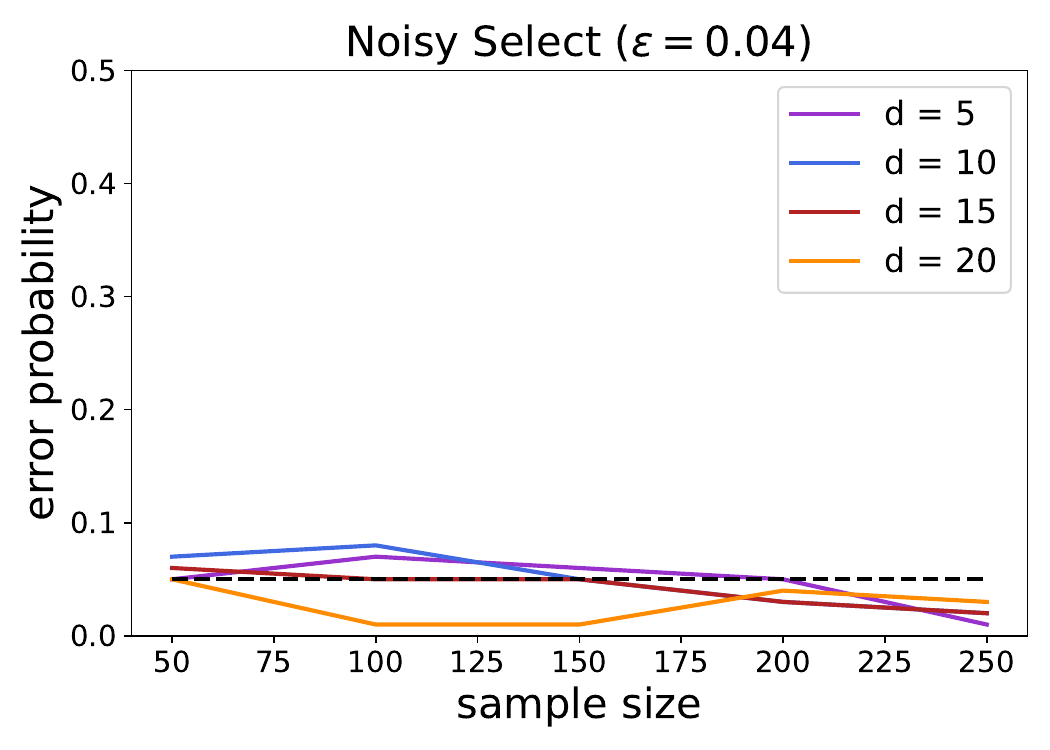}
        \caption{Probability of error after causal discovery via exact selection and \textsc{noisy-select} for varying $\epsilon$, $n$, and $d$. We measure error with respect to the true causal effect, rather than the projection effect.}
        \label{fig:true_effect_select}
\end{figure}

\begin{figure}[t]
     \centering
     \includegraphics[width=0.32\textwidth]{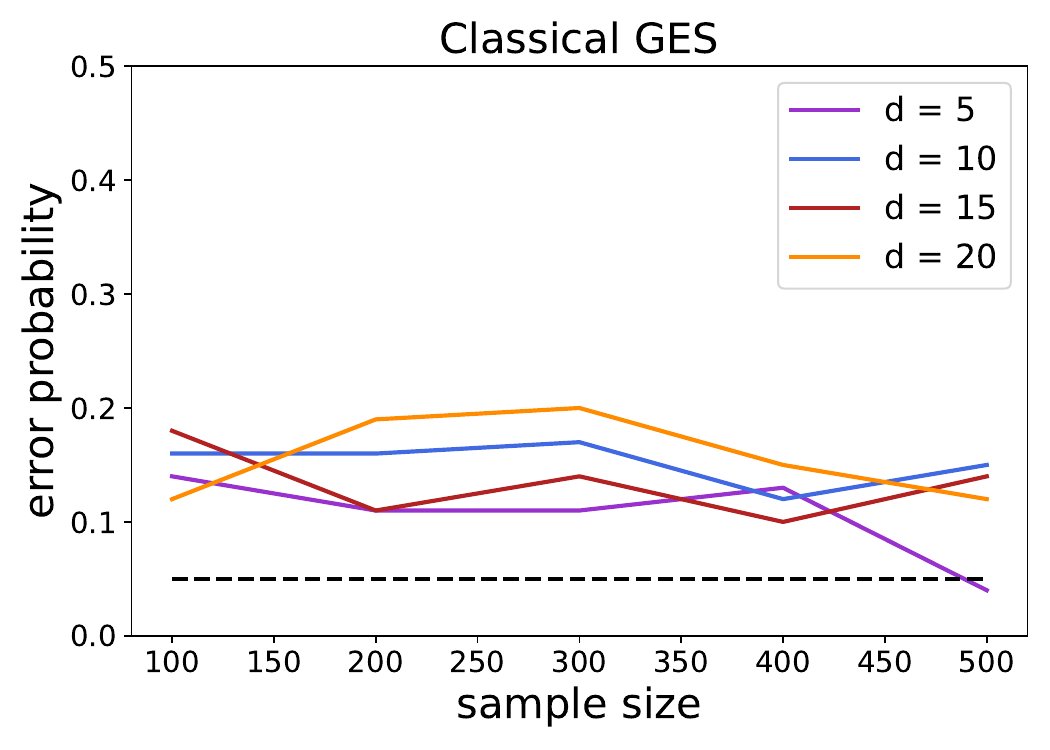}
     \includegraphics[width=0.32\textwidth]{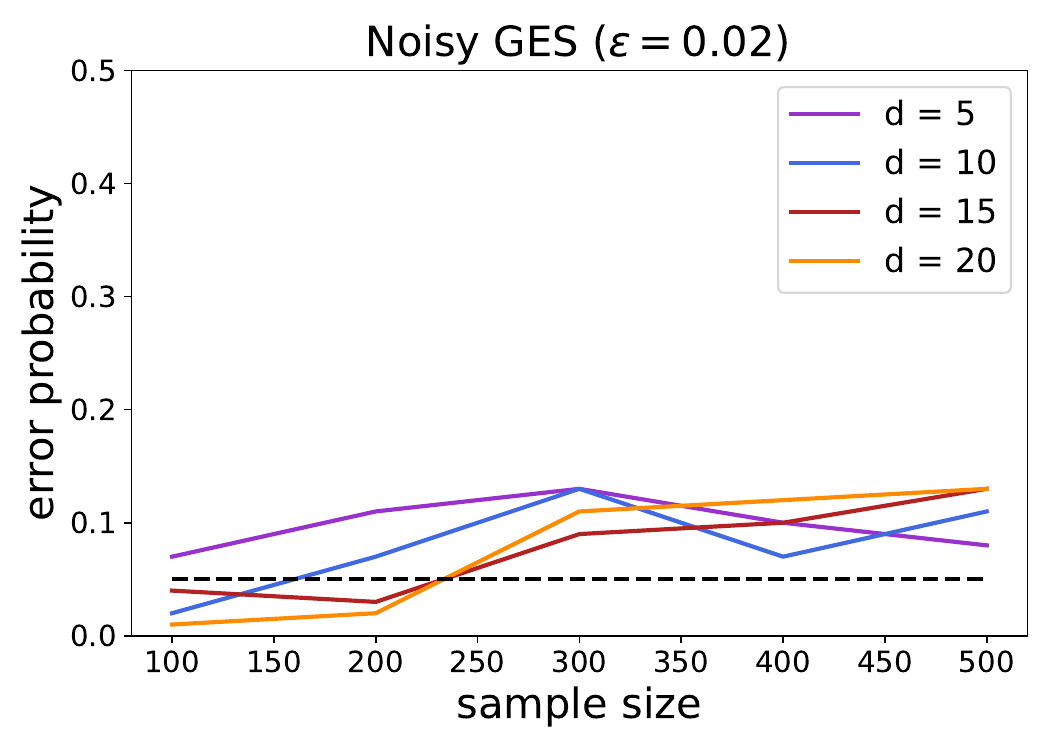}
    \includegraphics[width=0.32\textwidth]{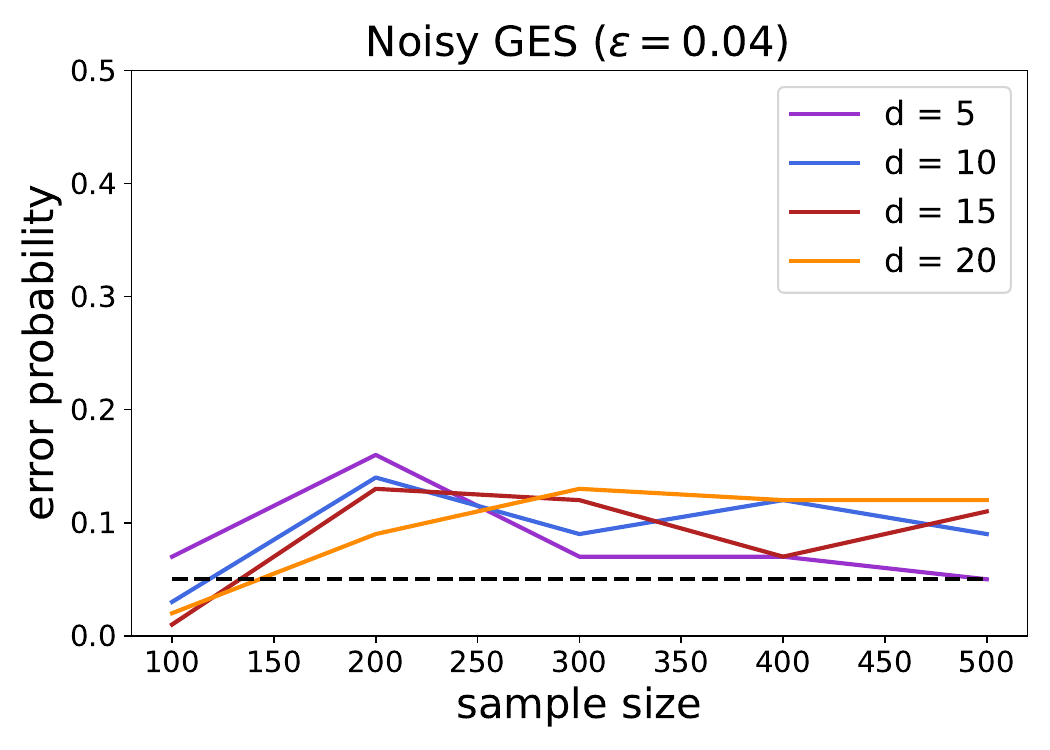}
        \caption{Probability of error after causal discovery via \textsc{noisy-ges} for varying $\epsilon$, $n$, and $d$. We measure error with respect to the true causal effect, rather than the projection effect.}
        \label{fig:true_effect_ges}
\end{figure}

\end{document}